\documentclass[reqno, 10pt,  letterpaper]{amsart}
\usepackage{statarb,nicefrac,booktabs}

\newcommand{\ccF}{{\mathscr F}}
\newcommand{\ccG}{{\mathscr G}}
\renewcommand{\F}{\ccF}
\renewcommand{\G}{\ccG}
\newcommand{\bbF}{\mathbb{F}}
\newcommand{\cI}{\mathcal{I}}
\newcommand{\bphi}{\boldsymbol \phi}
\newcommand{\bpsi}{\boldsymbol \psi}
\newcommand{\barA}{\overline{\text{SA}}}

\usepackage{stmaryrd}
\newcommand{\Ind}{\mathds 1}

\newcommand{\bx}{{\mathbf x}}
\newcommand{\bgamma}{{\boldsymbol \gamma}}

\title[Generalized statistical arbitrage concepts]{Generalized statistical arbitrage concepts and related gain strategies}
\author{Christian Rein}
\author{Ludger R\"uschendorf}
\address{Freiburg University, Dep. of Mathematics, Ernst-Zermelo Str. 1, 79104 Freiburg, Germany.}
	\email{ch.rein@gmx.net, ruschen@stochastik.uni-freiburg.de}
\author{Thorsten Schmidt}
\address{Freiburg Institute of Advanced Studies (FRIAS), Germany. 
 University of Strasbourg Institute for Advanced Study (USIAS), France. 
 University of Freiburg, Department of Mathematical Stochastics, Ernst-Zermelo-Str. 1, 79104 Freiburg, Germany}
 	\email{thorsten.schmidt@stochastik.uni-freiburg.de}
\date{\today}

\usepackage[colorlinks,urlcolor=red,citecolor=blue,linkcolor=red]{hyperref}
\usepackage[abs]{overpic}

\begin{document}

\maketitle

\begin{abstract}
Generalized statistical arbitrage concepts are introduced corresponding  to  trading strategies which yield positive gains \emph{on average} in a class of scenarios rather than almost surely. The relevant scenarios or market states are specified via an information system given by a   $\sigma$-algebra and so this notion contains classical arbitrage as a special case. It also covers the notion of \emph{statistical arbitrage} introduced in \cite{bondarenko2003statistical}. 

Relaxing these notions further we introduce generalized profitable strategies which include also static  or semi-static strategies. 
Under standard no-arbitrage  there may exist generalized gain strategies yielding positive gains on average under the specified scenarios. 

In the first part of the paper we characterize these generalized statistical no-arbitrage notions. In the second part of the paper we construct several profitable generalized strategies with respect to various choices of the information system. In particular, we consider several forms of embedded binomial strategies and follow-the-trend strategies as well as partition-type strategies. We study and compare their behaviour on simulated data. Additionally, we find good performance on market data of these simple strategies which makes them profitable candidates for real applications.
\end{abstract}

\section{Introduction}

Since the mid-1980s trading strategies which offer profits on average in comparison to little remaining risk have been implemented and analyzed. The starting point were  pairs trading strategies, see \cite{Gatev2006} for an historic account and further details. In this strategy  one trades two stocks whose prices have a high historic correlation and whose spread widened recently by buying the looser and shorting the winner. Many variants of this simple strategy followed, see \cite{Krauss2017} for a survey and a guide to the literature. This raised interest in a deeper theoretical understanding of these approaches. %

In this paper, we elaborate on the notion of \emph{statistical arbitrage} (SA)  introduced in \cite{bondarenko2003statistical}. 
The author considers a finite horizon market in order to restrict the class of admissible pricing rules. A trading strategy with zero initial cost is called statistical arbitrage  if 
\begin{enumerate}[(i)]
\item 
the expected payoff is positive and,
\item 
the \emph{conditional} expected payoff is non-negative in each final state of the economy.
\end{enumerate}

Unlike pure arbitrage strategies a statistical arbitrage can have negative payoffs provided the average payoff in each final state is non-negative. This concept supplements previous forms of restrictions like `good deals'  or opportunities with high Sharpe ratios or with high utility (see \cite{HansenJagannathan91},   \cite{CochraneSaaRequejo2000} and \cite{CernyHodges2002}) or `approximate arbitrage opportunities' and investment opportunities with a high gain-loss ratio (see \cite{BernardoLedoit2000}). All these restrictions lead to essential reductions of the pricing intervals.  

\cite{bondarenko2003statistical} discusses the concept of statistical arbitrage in connection with various forms of risk preferences, w.r.t. the solution of the joint hypothesis problem, for tests of the efficient market hypothesis (EMH) and the efficient learning market (ELM). The main economic assumption introduced by Bondarenko is the assumption that the pricing kernel is path independent, i.e.~it is a function depending only on the final state of the underlying price model but \emph{not} depending on the whole history.
This assumption implies that the payoff process deflated by the conditional risk neutral density of the final state is a martingale, i.e. has no systematic trend. 
The main result in \cite[Proposition 1]{bondarenko2003statistical} states that the existence of a path-independent pricing kernel is equivalent to the absence of SA strategies.

Following \cite{HoganJarrow2004}, another  strand of literature considers trading strategies which achieve positive gains on average together with vanishing risk in an asymptotic sense, see for example \cite{Elliott2005pairs,Avellaneda2010statistical}.

In Section 2 we generalize the concept of statistical arbitrage: starting from a general information system given by a $\sigma$-field $\ccG$, a \emph{statistical $\ccG$-arbitrage} is a trading strategy with positive expected gain conditionally on $\ccG$. The existence of a pricing measure with $\ccG$-measurable density implies absence of $\ccG$-arbitrage. 
Investigating in Section 3 in detail a class of trinomial models we find that the converse direction in Bondarenko's equivalence theorem is not valid in general. For two-period binomial models we fully characterize SA and construct statistical arbitrage strategies.  
In Section 4 we introduce generalized trading strategies including also static or semi-static strategies and derive various characterizations of the corresponding SA concepts; in particular we give conditions which imply equivalence results with the existence of $\ccG$-measurable pricing densities. 
In Section 5 we construct for discrete and continuous time models various SA-strategies, test them in several examples and give an application to market data.  
A basic class of strategies is obtained by embedding binomial trading strategies into the continuous time models using first-hitting times. Further classes are strategies induced by partitioning the path space and strategies which follow some trend in the data.

Several of theses strategies are examined and compared. As a result we obtain some  useful gain strategies and suggestions relevant for practical applications.

\section{Generalized gain strategies}
Consider a filtered probability space  $(\Omega, \ccF, P)$ with a filtration $\bbF=(\F_t)_{0 \leq t \leq T}$. The filtration is assumed to satisfy the usual conditions, \ie it is right continuous and $\F_0$ contains all null sets of $\F$: if $B \subset A \in \F$ and $P(A) = 0$ then $B \in \F_0$. We also suppose that  $\F=\F_T$. 

Following the classical approach to financial markets as for example in \cite{delbaen2006mathematics}, we consider a finite time horizon $T \in \N$.  The market itself is given by a $\R^{d+1}$-valued locally bounded non-negative semi-martingale $S=(S^0,\dots,S^d)$. The num\'eraire $S^0$ is set equal to one, such that the prices are considered as already discounted.

A \emph{dynamic trading strategy }$\phi$ is an $S$-integrable and predictable process such that the associated value process  $V=V(\phi)$ is given by
\begin{equation}
 V_t(\phi) = \int_0^t \phi_s \dS_s, \quad 0 \le t \le T.
\end{equation}
The trading strategy $\phi$ is called $a$-admissible if $\phi_0=0$ and $V_t(\phi) \geq -a$ for all $t \geq 0$. $\phi$ is called \emph{admissible} if it is admissible for some $a>0$.
We further assume that the market is free of arbitrage in the sense of \emph{no free lunch with vanishing risk} (NFLVR), which is equivalent to the existence of an equivalent local martingale measure $Q$, see \cite{delbaen2006mathematics}. Here, a measure $Q$ which is equivalent to $P$,  $Q\sim P$, such that $S$ is a $\bbF$-(local) martingale with respect to $Q$ is called equivalent (local) martingale measure, EMM (ELMM). Let $\Me$ denote the set of all equivalent local martingale measures.

A statistical arbitrage is a dynamic trading strategy which is \emph{on average} profitable, 
conditional on the final state of the economy $S_T$. More generally, we consider a general information system represented by a $\sigma$-field  $\G \subset \F_T$ and consider strategies which are on average profitable conditional on $\G$. For example, $\ccG$ could be generated by the event $\{S_T>K\}$, or the events $S_T \in K_i$, where $(K_i)_{i\in\cI}$ is a partition of $\R^d$, or by  $\{\max_{0 \le t \le T} S_t > K\}$. We call such strategies $\ccG$-arbitrage strategies. 
Sometimes we call a statistical $\G$-arbitrage strategy also a $\G$-profitable strategy or $\G$-arbitrage, for short. By $E$ we denote expectation with respect to the reference measure $P$.

\begin{definition}%
\label{Def.I.1}
  Let $\G \subseteq \F_T$ be a $\sigma$-algebra.
 An admissible dynamic trading strategy $\phi$ is called a \emph{statistical $\G$-arbitrage strategy}, if $V_T(\phi)\in L^1(P)$ and
 \begin{enumerate}[i)]
  \item $E[V_T(\phi) | \G] \geq 0, \quad P\text{-a.s.},$
  \item $E[V_T(\phi)] > 0$.
 \end{enumerate}
 Let 
 $$ \AG := \{\phi: \phi \text{ is a } \G\text{-arbitrage} \} $$ denote the set of all statistical $\G$-arbitrage strategies.
  The market model satisfies the condition of \emph{no statistical $\G$-arbitrage $\NAG$} if $$\AG = \emptyset. $$ 
\end{definition}

For $\G = \F_T$,  $\NAG$ is equivalent to the classical no-arbitrage condition (NA) since then $E[V_T(\phi)|\G]=V_T(\phi)$. Recall that NA is implied by NFLVR.
If $\G=\sigma(S_T)$, one recovers the notion of statistical arbitrage introduced in  \cite{bondarenko2003statistical}. 

A further interesting type of  examples is the case where 
  $\G=\sigma(\{ S_T \in K_i, i \in \cI\})$,  $\{K_i\}_{i \in \cI}$ being a partition of the state space, such that a statistical arbitrage offers a gain  in any  $\{ S_T \in K_i \}$ \emph{on average}, i.e.~$E[V_T(\phi)|S_T \in K_i] \ge  0$ for all $i \in \cI$.
Similarly one can also consider path-dependent strategies, like for example
$\G=\sigma(\{ \max_{0 \le t \le T}S_t \in K_i, i \in \cI\})$.

\begin{remark}[Relation to good-deal bounds]
The general approach to good-deal bounds  in  \cite{CernyHodges2002} allows to consider statistical arbitrages as a special case: indeed, if we define
$$ A= \{Z: E[Z|\ccG] \ge 0 \text{ and } E[Z]  >0 \}$$
as set of good deals then a statistical $\ccG$-arbitrage $\phi$ is a good deal strategy if $V_T(\phi) \in A$. 
The corresponding good-deal pricing bound is given by
$$ \pi(X) = \inf\{ x: \exists \phi \text{ admissible s.t. } X + x + V_T(\phi) \in A\}. $$
   \end{remark}
   
\begin{remark} We note some easy consequences of  Definition \ref{Def.I.1}.
 \begin{enumerate}[(i)]
  \item The tower property of conditional expectations  immediately yields  that  larger information systems $\G$ \,allow for less profitable $\G$-arbitrage strategies \ie  $\G_1 \subset \G_2$ implies that $ \text{SA}(\G_2) \subset \text{SA}(\G_1)$.
   As a consequence we get that in this case
  \begin{equation}\label{Eq.I.2}
   \text{NSA}(\G_1) \hspace*{5ex} \Rightarrow \hspace*{5ex} \text{NSA}(\G_2).
  \end{equation}
  \item  If $\G = \{\emptyset, \Omega\}$, then $\phi \in \AG$ iff $E_P [V_T(\phi)] > 0$.
 \end{enumerate}
\end{remark}

\section{On the statistical no-arbitrage notion}

  The notion of no statistical arbitrage is motivated by the question  whether it is possible to construct a trading strategy $\phi$ such that in any final state of the price process $S_T$ the trader gets a gain on average (\ie conditional on $\sigma(S_T)$). 

Proposition 1 in \cite{bondarenko2003statistical} states that (in discrete time),  NSA  is equivalent to the existence of an equivalent martingale measure $Q$ with path independent density $Z$, \ie 
\begin{equation}\label{Eq.I.3}
 \frac{dQ}{dP} = Z \in \sigma(S_T),
\end{equation}
where we use the notation $Z \in \sigma(S_T)$ for $Z$ being $\sigma(S_T)$-measurable.
We show in Section \ref{sec:counterexample}, that this equivalence needs additional assumptions which is one motivation of our work. In Section \ref{sec:statarbitrages} we explicitly construct statistical arbitrages whose study is the second motivation of our work.

On the other side, existence of an equivalent martingale measure with path independent density $Z$ implies that NSA holds without further assumptions. This also holds true for  the generalized notion $\NAG$, as we now show.

\begin{proposition}\label{Prop.I.1}
 If there exists $Q\in\Me$ such that $\frac{dQ}{dP}$ is $\G$-measurable, then $\NAG$ holds. 
\end{proposition}
\begin{proof}
 The proof follows from the Bayes-formula for conditional expectations. If  $Z = \frac{dQ}{dP} \in \G$, then for any $X \in L^1( P)$ it  holds that
 \begin{equation}\label{Eq.I.4}
  E_P[X \sd \G] = \frac{E_Q[XZ \sd \G]}{E_Q[Z \sd \G]} = E_Q[X \sd \G].
 \end{equation}
 If there would be a statistical arbitrage strategy $\phi$ with 
  $E_P[X \sd \G] \geq 0$ and $E_P[X] > 0$, where $X=V_T(\phi)\in L^1(P)$, then, by (\ref{Eq.I.4}),
  \begin{equation*}
   E_Q[X \sd \G] \geq 0, \hspace*{5ex} \quad Q\text{-a.s.}
  \end{equation*}

 Moreover, since $\phi$ is admissible, $V(\phi)$ is a $Q$-supermartingale by Fatou's lemma, and we obtain 
 that 
  \begin{equation}\label{supermartingaleproperty}
   E_Q[X] =  E_Q[V_T(\phi)] \le V_0(\phi)=0.
  \end{equation}
Hence, 
  \begin{equation*}
   0 = E_Q[X \sd \G] = E_P[X \sd \G]
  \end{equation*}
  in contradiction to $E_P[X] > 0$.
\end{proof}
\begin{remark}[Alternative admissible strategies]
An inspection of the proof, in particular Equation \eqref{supermartingaleproperty}, shows that the claim also holds when we consider as admissible such strategies $\phi$ for which $V(\phi)$ is a $Q$-martingale.\end{remark}

In the following we discuss whether also the converse direction in the Bondarenko result is true, \ie the question if no statistical arbitrage implies the existence of an equivalent martingale measure with path-independent density. Moreover we study the question  how statistical $\G$-arbitrage strategies can be constructed.

\subsection{Statistical arbitrage in trinomial models}\label{sec:trinomial}
In this section we consider a special one-dimensional trinomial model of the following type which we will call \emph{the trinomial model}. 
While the first step is binomial, the second time-step is trinomial. In this regard,
assume that $d=1$, $\Omega = \{\omega_1, \dots, \omega_6\}$ and $T=2$. Let $S_0 = s_0 \in \R_{\ge 0}$ and $S_1$ take the \emph{two} values $s_1^+$ and $s_1^-$ such that 
\begin{equation*}
 S_1(\omega_1) = S_1(\omega_2) = S_1(\omega_3)=s_1^+, \quad S_1(\omega_4) = S_1(\omega_5) = S_1(\omega_6)=s_1^-. 
\end{equation*}
The existence of an equivalent martingale measure $Q \sim P$ is equivalent to  $\Delta S_i = S_i - S_{i-1}$ taking positive as well as negative values in each sub-tree. 
For the first time step we assume without loss of generality that $s_1^+-s_0 >0 $ and $s_1^--s_0<0$.

For the second step we assume that the model takes the four values $s_2^{++}$, $s_2^{+-}$, $s_2^{--}$ and the top state $s_2^{\circ}$ with $s_2^{\circ}>s_2^{++}>s_2^{+-}>s_2^{--}>0$. While the $+/-$ states are reached by following a standard binomial, recombining two-period model, i.e.
\begin{align*}
S_2(\omega_2)=s_2^{++},\quad S_2(\omega_3)=S_2(\omega_5)=s_2^{+-}, \quad S_2(\omega_6)=s_2^{--}, 
\end{align*} 
the top state is reached by
$$ S_2(\omega_1)=S_2(\omega_4)=s_2^{\circ}. $$
We illustrate the scheme in Figure \ref{Fig.II.7}.

\begin{figure}[t]
  \centering
  \begin{tikzpicture}[>=stealth, sloped]
    \matrix (tree) [%
      matrix of nodes,
      minimum size=0.5cm,
      column sep=1.5cm,
      row sep=0.5cm,
    ] 
    {
      &     & $S_2(\omega_1) = S_2(\omega_4)$ &\\  
      &     &         &\\
      &         & $S_2(\omega_2)$     &\\
      & $S_1(\omega_1)$ &             &\\
        $S_0(\omega)$ &         & $S_2(\omega_3) = S_2(\omega_5)$   &\\
      & $S_1(\omega_4)$ &             &\\
      &         & $S_2(\omega_6)$       &\\
    };
    \draw[-] (tree-5-1)   -- (tree-4-2);%
    \draw[-] (tree-5-1)   -- (tree-6-2);
    
    \draw[-] (tree-4-2)   -- (tree-3-3);%
    \draw[-] (tree-4-2)   -- (tree-1-3);%
    \draw[-] (tree-4-2)   -- (tree-5-3);
    \draw[-] (tree-6-2)   -- (tree-5-3);%
    \draw[-] (tree-6-2)   -- (tree-1-3);    
    \draw[-] (tree-6-2)   -- (tree-7-3);

  \end{tikzpicture}
  \caption{The considered  trinomial model with $T=2$ time steps. The first step is binomial, the second step is also (recombining) binomial with an additional top state $\{\omega_1,\omega_4\}$. }\label{Fig.II.7}
\end{figure}
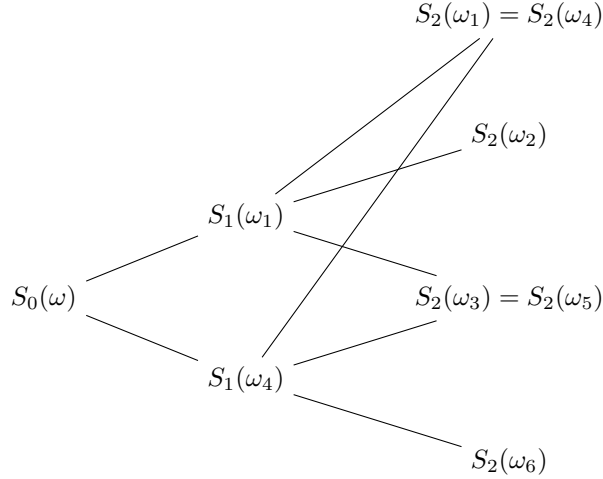

To ensure absence of arbitrage we assume that $s_2^{++}-s_1^+>0$, $s_1^-<s_2^{+-}<s_1^+$, $s_2^{--}-s_1^-<0$.
The gains from trading at time $2$ with a self-financing strategy $\phi$ are given by 
\begin{equation}\label{eq:vphi}
 V_2(\phi) = \phi_1 \Delta S_1 + \phi_2 \Delta S_2.
\end{equation}
While $\phi_1$ is constant since $\ccF_0=\{\emptyset,\Omega\}$, $\phi_2$ can take two different values which we denote by $ \phi^+_2$ and $\phi^-_2$ (taken in the states $\{\omega_1,\omega_2,\omega_3\}$ and $\{\omega_4,\omega_5,\omega_6\},$ respectively).  

Since $\G= \F_2 = \sigma(\{\omega_1, \omega_4\}, \{\omega_3, \omega_5\}, \{\omega_2\}, \{\omega_6\})$ the strategy $\phi$ is a statistical arbitrage if and only if
\begin{align}
 \phi_1 \Delta S_1(\omega_2) + \phi_2^+ \Delta S_2(\omega_2) &\geq 0 \label{CEeq1}, \\[1ex]
 \phi_1 \Delta S_1(\omega_6) + \phi_2^- \Delta S_2(\omega_6) &\geq 0 \label{CEeq2}, \\[1ex]
 \begin{split}\phi_1 \Delta S_1(\omega_1)P(\omega_1) + \phi_2^+ \Delta S_2(\omega_1) P(\omega_1) & \\
 + \phi_1 \Delta S_1(\omega_4)P(\omega_4) + \phi_2^- \Delta S_2(\omega_4) P(\omega_4) &\geq 0, \end{split} \label{CEeq3} \\[1ex]
 \begin{split}\phi_1 \Delta S_1(\omega_3)P(\omega_3) + \phi_2^+ \Delta S_2(\omega_3) P(\omega_3) & \\
 + \phi_1 \Delta S_1(\omega_5)P(\omega_5) + \phi_2^- \Delta S_2(\omega_5) P(\omega_5) &\geq 0, \end{split} \label{CEeq4}
\end{align}
and, in addition, at least one of the inequalities is strict. \\

Moreover, if we consider an equivalent martingale measure $Q$ then the density $Z$ is path-independent if and only if $Z(\omega_1) = Z(\omega_4)$ and $Z(\omega_3) = Z(\omega_5)$. As a next step we  establish a criterion for our model to be free of statistical arbitrage. Denote 
  \begin{eqnarray*}
   \Gamma_1 &=& \frac{-\Delta S_1(\omega_5) + \Delta S_2(\omega_5) \frac{\Delta S_1(\omega_6)}{\Delta S_2(\omega_6)}}{\Delta S_1(\omega_3) - \Delta S_2(\omega_3)\frac{\Delta S_1(\omega_2)}{\Delta S_2(\omega_2)}},\\
   \Gamma_2 &=& \frac{\frac{\Delta S_1(\omega_6)}{\Delta S_2(\omega_6)} (\Delta S_2(\omega_4) + \Delta S_2(\omega_5)) - \Delta S_1(\omega_4) - \Delta S_1(\omega_5)}{\Delta S_1(\omega_3) - \Delta S_1(\omega_1)\frac{\Delta S_2(\omega_3)}{\Delta S_2(\omega_1)}}.
  \end{eqnarray*}

\begin{lemma}\label{CElemma1}
 Let $\nu_1 := \frac{P(\omega_1)}{P(\omega_4)}$ and $\nu_2 := \frac{P(\omega_3)}{P(\omega_5)}$. In the trinomial model there is no statistical arbitrage if $\nu_1 = -\frac{\Delta S_2(\omega_3)}{\Delta S_2(\omega_1)} \nu_2$ and if it holds that
 \begin{equation}\label{CEeq8}
  \Gamma_1 < \nu_2 \leq \Gamma_2.
 \end{equation}
\end{lemma}

The proof is relegated to the appendix.

\subsection{A counter example}\label{sec:counterexample}

In the following we use Lemma \ref{CElemma1} to show that Proposition 1 in \cite{bondarenko2003statistical} is not valid without additional conditions. 
Consider the  (incomplete) trinomial model  specified in Figure \ref{Fig.1}. 
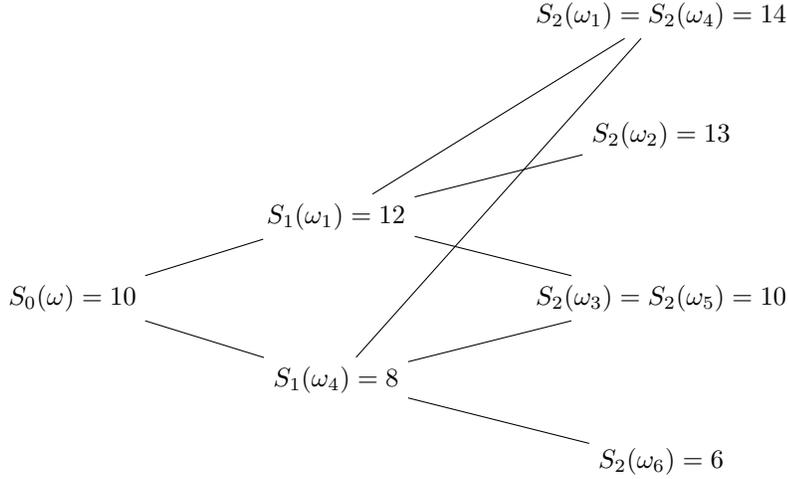
\begin{figure}[t]
  \centering
  \begin{tikzpicture}[>=stealth, sloped]
    \matrix (tree) [%
      matrix of nodes,
      minimum size=0.5cm,
      column sep=1.5cm,
      row sep=0.5cm,
    ] 
    {
        &     & $S_2(\omega_1) = S_2(\omega_4) = 14$  &\\  
        &     &         &\\
        &         & $S_2(\omega_2) = 13$      &\\
        & $S_1(\omega_1) = 12$  &             &\\
        $S_0(\omega) = 10$  &         & $S_2(\omega_3) = S_2(\omega_5) = 10$  &\\
        & $S_1(\omega_4) = 8$ &             &\\
        &         & $S_2(\omega_6) = 6$       &\\
    };
    \draw[-] (tree-5-1)   -- (tree-4-2);%
    \draw[-] (tree-5-1)   -- (tree-6-2);
    
    \draw[-] (tree-4-2)   -- (tree-3-3);%
    \draw[-] (tree-4-2)   -- (tree-1-3);%
    \draw[-] (tree-4-2)   -- (tree-5-3);
    \draw[-] (tree-6-2)   -- (tree-5-3);%
    \draw[-] (tree-6-2)   -- (tree-1-3);    
    \draw[-] (tree-6-2)   -- (tree-7-3);

  \end{tikzpicture}
  \caption{An explicit trinomial model with $T=2$ time steps}\label{Fig.1}
\end{figure}

It is easy to check that the equivalent martingale measures $Q$ specified by $ q=(Q(\omega_1),\dots,Q(\omega_6))$ are given by the set
\begin{align*}
 \mathcal{Q} = \Big\{ q \in \R^6 \, \Big| \, & q_1 = -\frac{3}{4} q_2 + \frac{1}{4}, q_3 = -\frac{1}{4}q_2 + \frac{1}{4}, q_4 = q_6 - \frac{1}{4}, q_5 = -2 q_6 + \frac{3}{4}, \\ 
 &\text{ where } q_2 \in \Big(\frac{1}{3},1\Big), q_6 \in \Big(\frac{1}{4}, \frac{3}{8}\Big) \Big\}.
\end{align*}
Furthermore consider the underlying measure $P$  uniquely specified by the vector $p=(P(\omega_1),\dots,P(\omega_6))$ given by
\begin{align*}
  p = (0.15, 0.2, 0.3, 0.05, 0.1, 0.2).
\end{align*}
We compute $\nu_1 = \frac{p_1}{p_4} = 3$ and $\nu_2 = \frac{p_3}{p_5} = 3$. Then
\begin{align*}
 \Gamma_2 &= \frac{\frac{\Delta S_1(\omega_6)}{\Delta S_2(\omega_6)} (\Delta S_2(\omega_4) + \Delta S_2(\omega_5)) - \Delta S_1(\omega_4) - \Delta S_1(\omega_5)}{\Delta S_1(\omega_3) - \Delta S_1(\omega_1)\frac{\Delta S_2(\omega_3)}{\Delta S_2(\omega_1)}} = 3 = \nu_2, \\
 \Gamma_1 &= \frac{-\Delta S_1(\omega_5) + \Delta S_2(\omega_5) \frac{\Delta S_1(\omega_6)}{\Delta S_2(\omega_6)}}{\Delta S_1(\omega_3) - \Delta S_2(\omega_3)\frac{\Delta S_1(\omega_2)}{\Delta S_2(\omega_2)}} = \frac{2}{3} < \nu_2
\end{align*}
and
\begin{equation*}
 \nu_1 = -\frac{\Delta S_2(\omega_3)}{\Delta S_2(\omega_1)} \nu_2 = \nu_2 = 3 = \frac{p_1}{p_4}.
\end{equation*}
According to Lemma \ref{CElemma1} there is no statistical arbitrage in the stated example. But, on the other hand, there is no path independent density in this case because if there would be a path independent density, \ie a density $Z$ with $Z(\omega_1) = Z(\omega_4)$ and $Z(\omega_3) = Z(\omega_5)$, there would exist an equivalent martingale measure $Q$ fulfilling the conditions
\begin{equation}\label{CEeq7}
 \frac{q_1}{q_4} = \frac{p_1}{p_4} = 3 \hspace*{5ex} \text{ and } \hspace*{5ex} \frac{q_3}{q_5} = \frac{p_3}{p_5} = 3.
\end{equation}
But the only $q \geq 0$ fulfilling (\ref{CEeq7}) is $q = (\frac{1}{4}, 0, \frac{1}{4}, \frac{1}{12}, \frac{1}{12}, \frac{1}{3})$ which is not an element of $\mathcal Q$. 

This example shows that Proposition 1 in \cite{bondarenko2003statistical} needs additional assumptions: indeed, we have shown that there does not exist a statistical arbitrage and at the same time there is no path-independent density. In Section  \ref{Section3} we study this topic in more detail.

\subsection{Statistical arbitrage strategies in binomial models}\label{sec:statarbitrages}
In this section we propose a method to construct statistical arbitrage strategies in binomial models.

Consider the following recombining two-period binomial model:  assume that $\Omega = \{\omega_1, \dots, \omega_4\}$ and $T=2$. Let $S_0 = s_0 >0$  and let $  S_1(\omega_1) = S_1(\omega_2)=s^+$, and $S_1(\omega_3) = S_1(\omega_4)=s^-$ as well as $s^{++}=S_2(\omega_1)$, $s^{+-}=S_2(\omega_2)=S_2(\omega_3)$, and $s^{--}=S_2(\omega_4)$. This model is illustrated in Figure \ref{Fig.II.8}.

Absence of arbitrage is equivalent to  $\Delta S_i$, $i=1,2$ taking positive as well as negative values. We assume without loss of generality that $s^+>s_0,$ $s^-<s_0$, and $s^{++}>s^+$, $s^-<s^{+-}<s^+$, and $s^{--}<s^-$ \ie we consider binomial models as presented in Figure \ref{Fig.II.8}.
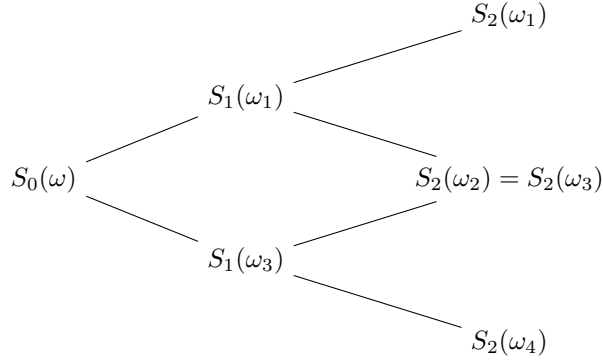
\begin{figure}[t]
  \centering
  \begin{tikzpicture}[>=stealth, sloped]
    \matrix (tree) [%
      matrix of nodes,
      minimum size=0.5cm,
      column sep=1.5cm,
      row sep=0.5cm,
    ] 
    {
      &         & $S_2(\omega_1)$     &\\
      & $S_1(\omega_1)$ &             &\\
        $S_0(\omega)$ &         & $S_2(\omega_2) = S_2(\omega_3)$   &\\
      & $S_1(\omega_3)$ &             &\\
      &         & $S_2(\omega_4)$       &\\
    };
    \draw[-] (tree-3-1)   -- (tree-2-2);%
    \draw[-] (tree-3-1)   -- (tree-4-2);
    
    \draw[-] (tree-2-2)   -- (tree-1-3);%
    \draw[-] (tree-2-2)   -- (tree-3-3);%
    \draw[-] (tree-4-2)   -- (tree-3-3);
    \draw[-] (tree-4-2)   -- (tree-5-3);%
  \end{tikzpicture}
  \caption{The considered recombining binomial model with two periods.}\label{Fig.II.8}
\end{figure}
Gains from trading are again given by \eqref{eq:vphi}. Also $\phi_1$ is constant and $\phi_2$ can take the two values $\{\phi_2^+,\phi_2^-\}$. 
As in Equations \eqref{CEeq1} - \eqref{CEeq4}, $\phi$ is a statistical arbitrage, iff 
\begin{align}
 \phi_1 \Delta S_1(\omega_1) + \phi_2^+ \Delta S_2(\omega_1) &\geq 0 \label{Eq.II.3.1} \\[1ex]
 \phi_1 \Delta S_1(\omega_4) + \phi_2^- \Delta S_2(\omega_4) &\geq 0 \label{Eq.II.3.2} \\[1ex]
 \begin{split}\phi_1 \Delta S_1(\omega_2)P(\omega_2) + \phi_2^+ \Delta S_2(\omega_2) P(\omega_2) & \\
 + \phi_1 \Delta S_1(\omega_3)P(\omega_3) + \phi_2^- \Delta S_2(\omega_3) P(\omega_3) &\geq 0 \end{split}\label{Eq.II.3.3}  
\end{align}
and at least one of the inequalities is strict. 
Moreover, the density $Z$ is path-independent if and only if $Z(\omega_2) = Z(\omega_3)$.
Equations (\ref{Eq.II.3.1}) - (\ref{Eq.II.3.3}) are equivalent to $A \bphi \geq 0 $, $\bphi=(\phi_1,\phi_2^+,\phi_2^-)^\top$ with 
  \begin{equation}\label{Eq.II.3.7}
  A = 
  \begin{pmatrix}
  \Delta S_1(\omega_1)        & \Delta S_2(\omega_1)    & 0 \\
  \Delta S_1(\omega_4)        & 0       & \Delta S_2(\omega_4)\\
  q \Delta S_1(\omega_2) + \Delta S_1(\omega_3) & q \Delta S_2(\omega_2)  & \Delta S_2(\omega_3)
  \end{pmatrix},
 \end{equation}
where $q = \frac{P(\omega_2)}{P(\omega_3)}$.

\begin{proposition}\label{Lem.II.1}
 In the recombining two-period binomial model NSA holds if and only if $\det(A) = 0$.
\end{proposition}

The proof is relegated to the appendix.

\begin{remark}\label{Rem.II.1}
 It turns out that in the binomial model above NSA is equivalent to existence of a path-independent density: indeed, 
  the unique equivalent martingale measure is given by the vector $B^{-1} (q_1, \dots ,q_4)$ with 
 \begin{align}
  q_1 &= \Delta S_2(\omega_2)\big(\Delta S_1(\omega_3)\Delta S_2(\omega_4) - \Delta S_1(\omega_4)\Delta S_2(\omega_3)\big), \nonumber \\
  q_2 &= - \Delta S_2(\omega_1)\big(\Delta S_1(\omega_3)\Delta S_2(\omega_4) - \Delta S_1(\omega_4)\Delta S_2(\omega_3)\big), \label{Eq.II.3.5} \\
  q_3 &= - \Delta S_2(\omega_4)\big(\Delta S_1(\omega_1)\Delta S_2(\omega_2) - \Delta S_1(\omega_2)\Delta S_2(\omega_1)\big), \label{Eq.II.3.6} \\
  q_4 &= \Delta S_2(\omega_3)\big(\Delta S_1(\omega_1)\Delta S_2(\omega_2) - \Delta S_1(\omega_2)\Delta S_2(\omega_1)\big) \nonumber
 \end{align}
 and 
 \begin{align*}
 B = & \Delta S_2(\omega_2)\Big(\big(\Delta S_1(\omega_3) - \Delta S_1(\omega_1)\big)\Delta S_2(\omega_4) + \big(\Delta S_1(\omega_1)-\Delta S_1(\omega_4)\big)\Delta S_2(\omega_3)\Big) 
 \\&+ \Delta S_2(\omega_1)\Big(\big(\Delta S_1(\omega_2) - \Delta S_1(\omega_3)\big)\Delta S_2(\omega_4) + \big(\Delta S_1(\omega_4) - \Delta S_1(\omega_2)\big)\Delta S_2(\omega_3)\Big).
 \end{align*}
 Proposition \ref{Lem.II.1} yields that NSA holds iff $\det(A) = 0$, which is   according to Equation (\ref{Eq.II.3.9}) equivalent to
 \begin{equation}\label{def:qtilde}
  \frac{P(\omega_2)}{P(\omega_3)} = \frac{\Delta S_2(\omega_1)(\Delta S_1(\omega_3)\Delta S_2(\omega_4) - \Delta S_1(\omega_4)\Delta S_2(\omega_3))}{\Delta S_2(\omega_4)(\Delta S_1(\omega_1)\Delta S_2(\omega_2) - \Delta S_1(\omega_2)\Delta S_2(\omega_1))} =: \tilde q.
 \end{equation}
 Using (\ref{Eq.II.3.5}) and (\ref{Eq.II.3.6}) we obtain from $\det(A) = 0$ that
 \begin{equation*}
  \frac{dQ(\omega_2)}{dQ(\omega_3)} = \tilde q = \frac{dP(\omega_2)}{dP(\omega_3)},
 \end{equation*}
 which means that NSA is equivalent to the existence of a path-independent density.
\end{remark}
The question now is what path properties imply absence of statistical arbitrage opportunities. 
\begin{lemma}\label{Lem.II.2}
 In the recombining two-period binomial model there exists a statistical arbitrage if and only if  
 \begin{equation}\label{NSArecbinom}
 \frac{P(\omega_2)}{P(\omega_3)} \neq \tilde q.
 \end{equation}
\end{lemma}
\begin{proof}
 To have the possibility of statistical arbitrage we know from Proposition \ref{Lem.II.1} that we need $\det(A) \neq 0$ which is, according to Remark \ref{Rem.II.1}, equivalent to $\frac{P(\omega_2)}{P(\omega_3)} \neq \tilde q$.
\end{proof}

The following Lemma explicitly describes the statistical arbitrages in terms of the vector $\bphi=(\phi_1,\phi_2^+,\phi_2^-)$
\begin{lemma}\label{Lem.II.3}
In the recombining two-period binomial model with statistical arbitrage, 
$\bphi = \tfrac 1 D  (\xi^1, \xi^2, \xi^3)$ with
 \small{\begin{align*}
  \xi^1 &= \big(q \Delta S_2(\omega_2) - \Delta S_2(\omega_1)\big)\Delta S_2(\omega_4) + \Delta S_2(\omega_1) \Delta S_2(\omega_3),\\
  \xi^2 &= -\big(\Delta S_1(\omega_3) + q \Delta S_1(\omega_2) - \Delta S_1(\omega_1)\big)\Delta S_2(\omega_4) - \big(\Delta S_1(\omega_1) - \Delta S_1(\omega_4)\big)\Delta S_2(\omega_3),\\
  \xi^3 &= -\big(q \Delta S_1(\omega_4) - q \Delta S_1(\omega_1)\big)\Delta S_2(\omega_2) - \big(-\Delta S_1(\omega_4) + \Delta S_1(\omega_3) + q \Delta S_1(\omega_2)\big)\Delta S_2(\omega_1),
 \intertext{$q = \frac{P(\omega_2)}{P(\omega_3)}$, and}
  D & = \Big(q\Delta S_1(\omega_1)\Delta S_2(\omega_2) + \big(-\Delta S_1(\omega_3) - q \Delta S_1(\omega_2)\big)\Delta S_2(\omega_1)\Big)\Delta S_2(\omega_4)\\ 
  &+ \Delta S_1(\omega_4) \Delta S_2(\omega_1) \Delta S_2(\omega_3)
 \end{align*}}
 is a statistical arbitrage.
\end{lemma}

\begin{proof}
 If $\frac{P(\omega_2)}{P(\omega_3)} \neq \tilde q$ we have statistical arbitrage according to Lemma \ref{Lem.II.2} and the determinant of the matrix $A$ in (\ref{Eq.II.3.7}) is not equal to zero according to Proposition \ref{Lem.II.1}. In this case the matrix $A$ is invertible. Hence, $\bphi =  A^{-1}\mathds{1}$ is  a statistical arbitrage and it is easily verified that $\bphi = \tfrac 1 D (\xi^1,\xi^2,\xi^3)$.
\end{proof}

In Section \ref{StatArbForDiffProcesses} we will use this information and propose a dynamic trading strategy exploiting statistical arbitrages with the results of this section.

\subsection{Risk of statistical arbitrages}\label{sec:3.4}
The word \emph{arbitrage} might be misleading on the riskiness of statistical arbitrages, because in the classical sense, an arbitrage is a strategy without risk. This is of course \emph{not} the case for statistical arbitrages (or the following generalizations of this concept). Since we consider arbitrage-free markets, all gains come with a certain risk and, higher profits are associated with higher risk.
This is confirmed by our simulation results in Section \ref{StatArbForDiffProcesses}.

As a simple example consider the case where $\Delta_iS(\omega_j)\in \{5,-5\}$, i.e.~the stock either rises by 5 or falls by 5. In addition, assume that $q=P(\omega_2)/P(\omega_3)=1.2$. Then, using Equation \eqref{Eq.II.3.7} it is not difficult to compute $\bphi =  A^{-1}\mathbf{1} = (1.6, -1.4, -1.8)^\top$. From this strategy we obtain that the gains at time 2, given by
$$ G_2(\omega) = \phi_1(\omega) \Delta S_1(\omega) + \phi_2(\omega) \Delta S_2(\omega), $$
yield $G_2(\omega_1)=G_2(\omega_4)=1$, corresponding to \eqref{Eq.II.3.1} and \eqref{Eq.II.3.2}. 
In addition, we obtain that $G_2(\omega_2)= 15$ and $G_2(\omega_3)= -17$. If we assume that $P(\omega_2)=0.3$ we obtain that the average expected gain on $\{\omega_2,\omega_3\}$ computes to
\begin{align}
  P(\omega_2) G_2(\omega_2) + P(\omega_3) G_3(\omega_3) &=
  0.3 \cdot 15 + 0.25 \cdot (-17) = 0.25 \ge 0,
\end{align}
such that the strategy is indeed a statistical arbitrage. While the (average) gains in the three relevant scenarios are $1,0.25,1$, the possible loss in scenario $\omega_3$ is equal to $-17$, which is attained with probability $0.25$, clearly pointing out the riskiness of the strategy. 

To exploit the averaging property of \emph{statistical} arbitrage, we repeat this strategy in the following until we first record a positive P\&L. 
These considerations show clearly, that a risk analysis of the implemented strategy is very important.

\section{Generalized $\G$-arbitrage strategies}\label{Section3}
In connection with improvement procedures for payoffs we consider any static or semi-static payoff $X \in L^1(P)$  as a generalized strategy. 
This leads to the following notion of generalized statistical $\G$-arbitrage strategies and the corresponding notion of generalized statistical $\G$-arbitrage. 
This concept was used in several papers dealing with improvement procedures of financial contracts, see for example \cite{kassberger2017additive}.
We denote by $L^1(P,Q):=L^1(P)\cap L^1(Q)$ the set of random variables which are integrable with respect to $P$ and $Q$.

\begin{definition}%
\label{Def.III.1}
 Let $\G \subseteq \F$ be a $\sigma$-algebra.
The set of generalized statistical $\G$-arbitrage-strategies with respect to $Q\in \Me $ is defined as
  \begin{equation*}
   \barA(Q,\G) := \{X \in L^1(P,Q): E_Q[X] = 0, ~ E_P[X | \G] \geq 0 ~ P\text{-a.s. and } E_P[X] > 0\}
  \end{equation*}
 The market satisfies  $\NAQq$, the condition of \emph{no generalized statistical $\G$-arbitrage} with respect to $Q$,  if 
     $$ \barA(Q, \G) = \emptyset. $$
\end{definition}

 We aim at studying under which conditions there exist generalized statistical $\G$-arbitrages and  to describe connections between  $\NAQq$ and $\NAG$. The following  result in \cite{kassberger2017additive}, Proposition 6, characterizes the generalized \NAQq-condition by showing that in fact this notion is equivalent to  $\G$-measurability of $dZ=\frac{dQ}{dP}$.

\begin{proposition}%
\label{Prop.III.1}
 Let $Q \in \Me$. Then  \NAQq is equivalent to the existence of a $\G$-measurable version of the Radon-Nikodym derivative $Z = \frac{dQ}{dP}$.
\end{proposition}
The proof of this result is achieved by Jensen's inequality and using as  candidate of a generalized $\G$-arbitrage 
\begin{equation}\label{Eq.III.1}
 X = \frac{E[Z \sd \G]}{Z} - 1 \ge -1.
\end{equation}
Equation \eqref{Eq.III.1} also shows that the statistical arbitrage, if it exists, may be chosen bounded from below. 

One consequence of this characterization result is the characterization of  $\NAG$  for the case of complete market models. Recall that the Radon-Nikodym derivative $Z=\frac{dQ}{dP}$ is path-independent, iff $Z$ is $\sigma(S_T)$-measurable.

 A financial market is called \emph{complete}, if every contingent claim is attainable, i.e.~for every $\F$-measurable random variable $X$ bounded from below, we find an admissible self-financing trading strategy $\phi$, such that $x+V_T(\phi)=X$. 
 This is implied by the assumption that $\Me=\{Q\}$: indeed, under this assumption, Theorem 16 in \cite{DelbaenSchachermayer1995b} yields that any $X\in L^1(Q)$, bounded from below, is hedgeable and hence attainable.

\begin{proposition}\label{Prop.III.3}
  Assume that $\Me=\{Q\}$. Then  NSA$(\ccG)$ holds if and only  if $\frac{dQ}{dP}$ is $\ccG$-measurable. 
\end{proposition}
\begin{proof}
    We first show that existence of a $\ccG$-measurable $Q\in \Me$  implies NSA$(\ccG)$: choose $Q \in \Me$, such that $Z =\frac{dQ}{dP}$ is $\ccG$-measurable. Then  NSA$(\ccG)$  follows as in the proof of Proposition \ref{Prop.I.1}.

  For the converse direction assume   that $Z$ is not $\G$-measurable.  By Proposition \ref{Prop.III.1} it follows that there exists a generalized $\ccG$-arbitrage, i.e.~an $X\in L^1(P,Q)$ with $E_Q[X]=0$, $E_P[X \sd \G] \geq 0$ and $E_P[X] >0$.  
      As remarked above, $X$ can be chosen bounded from below. Hence, Theorem 16 in  \cite{DelbaenSchachermayer1995b} yields existence of an admissible self-financing trading strategy $\phi$, such that $x+V_T(\phi)=X$. Moreover, the superhedging duality, i.e. Theorem 9 in  \cite{DelbaenSchachermayer1995b} implies that $x=E_Q[X]=0$, and hence $\phi$ is a $\ccG$-arbitrage.
  This is a contradiction and the claim follows.
\end{proof}

In particular this result implies that Proposition 1 in \cite{bondarenko2003statistical} gives a correct characterization of NSA for complete markets.

\begin{example}[Statistical arbitrage for diffusions]\label{Ex.1}
  This example discusses the consequences of  Proposition \ref{Prop.III.1} and Proposition \ref{Prop.III.3} in the case of a diffusion model. Let  $S$ be a one-dimensional diffusion process  satisfying 
 \begin{equation}\label{Eq.III.3}
  dS_t = a_t \dt + b_t \dB_t, \quad 0 \leq t \leq T,
 \end{equation}
  where $B_t$ is a $P$-Brownian motion, $a$ and $b$ are progressively measurable such that $P(\int_0^T |a_s|ds < \infty )=1$ and $P(\int_0^T b_s^2 ds < \infty) = 1$. Assume further that $b>0$ $dt$-almost surely that the Novikov-condition is satisfied, \ie
  \begin{equation*}
   E\left[\exp \Big(\frac{1}{2} \int_0^T \frac{a^2_s}{b^2_s} \ds \Big)\right] < \infty.
  \end{equation*}
  Then this model is complete and by Girsanov's theorem has a unique equivalent local martingale measure $Q$ with Radon-Nikodym derivative
  \begin{equation}\label{Eq.III.4}
   Z_T = \exp\left(-\int_0^T \frac{a_t}{b_t} \dB_t - \frac{1}{2}\int_0^T \frac{a_t^2}{b_t^2} \dt \right).
  \end{equation}
  If $a_t/b_t = c$ $dt$-almost surely, then  we obtain from Proposition \ref{Prop.III.3} that there are no statistical arbitrage opportunities.
  This holds in particular when $a_t=a_0$ and $b_t=b_0$, $0 \le t \le T,$ \ie in the case of constant drift and volatility (the Black-Scholes model). On the other side, the diffusion model allows for statistical arbitrage  except for the case that $ (a_t / b_t)$ is constant $dt$-almost surely. A comparable result was obtained in \cite{Goncu2015} when studying the concept of statistical arbitrage introduced in \cite{HoganJarrow2004} in the Black-Scholes model.
\end{example}
The following definition introduces the generalized $\G$-no-arbitrage condition without dependence  on a specific pricing measure $Q$.

\begin{definition}%
\label{Def.III.2}
 Let $\G \subseteq \F$ be a $\sigma$-algebra. The  set of \emph{generalized statistical $\G$-arbitrage-strategies} is defined as
  \begin{equation*}
   \barA(\G) := \{X \in L^1(P): \sup_{Q \in \Me} E_Q[X] \le 0, ~ E_P[X | \G] \geq 0 ~ P\text{-a.s. and } E_P[X] > 0\}.
  \end{equation*}
  The market satisfies  $\NAq$, i.e.~\emph{no  generalized statistical $\G$-arbitrage}, if 
       $$ \barA(\G) = \emptyset. $$
 \end{definition}
  Note that the definition defines a \emph{generalized statistical $\G$-arbitrage} as a random variable $X \in L^1(P)$, such that $\sup_{Q \in \Me} E_Q[X]\le 0$,  $E_P[X | \G] \ge 0$, $P$-almost surely, and
  $E_P[X]>0$.
  In this sense, the strategies in $\Aq$ are generalized statistical $\G$-arbitrage-strategies under any choice of the pricing measure $Q$. Our next step is to establish a relation between $\G$-arbitrages and generalized $\G$-arbitrages. Note that the connection to trading strategies in a continuous-time setting requires, as usual, to allow that $\sup_{Q \in \Me} E_Q[X]$ may be negative, while for the definition of $\barA(Q,\G)$ we were able to consider $E_Q[X]=0$. The precise reasoning for this is becoming clear in the proof of the next proposition.
  
  We use the concept of \textit{No Free Lunch with Vanishing Risk} (NFLVR), which is a mild strengthening of the no-arbitrage concept, and refer to \cite{delbaen1994general} for definition and further reading. According to the results in this article we require in the following that $S$ is locally bounded, i.e.~there exists a sequence of stopping times $(T_n)_{n \ge 1}$ tending to $\infty$ a.s.~and a sequence $(K_n)_{n \ge 1}$  of positive constants, such that $|S \Ind_{\llbracket 0, T_n \rrbracket}| < K_n$, $ n \ge 1$.

The set of generalized $\G$-arbitrage strategies restricted to  claims  bounded from below is denoted by
\begin{align*}
    \overline{\text{SA}}_b(\G) :=\barA(\G) \cap  \{X \in L^1(P):& \exists \, a \in \R \text{ such that }X \geq - a\} .
\end{align*}

\begin{proposition}\label{Prop.III.2}
 Assume that $S$ satisfies (NFLVR). Then
 $$ \overline{\text{NSA}}_b(\G)  \Leftrightarrow \NAG. $$
\end{proposition}
\begin{proof}
 We first show that every $\G$-arbitrage strategy is a generalized $\G$-arbitrage strategy: consider $\phi \in \AG$, i.\,e. $E[V_T(\phi) \sd \G] \geq 0$ and $E[V_T(\phi)]>0$. . By the superreplication duality, Theorem 9 in \cite{delbaen1995no}, it holds that
  \begin{equation*}
  \sup_{Q \in \Me} E_Q [V_T(\phi)] = \inf\{x \sd \exists \text{ admissible }\tilde{\phi}, x + V_T(\tilde{\phi}) \geq V_T(\phi) \}.
 \end{equation*}
 Choosing $\tilde{\phi} = \phi$ it follows $\sup_{Q \in \Me} E_Q V_T(\phi) \leq 0$. Note that in addition, admissibility of $\phi$ implies that  $V_T(\phi) $ is bounded from below and so $V_T(\phi)\in \overline{A}_b(\G)$.
 
 For the reverse implication we have, again by the superreplication duality, for $X \in \overline{A}_b(\G)$ that 
   \begin{equation*}
  0 \ge \hspace{-1ex} \sup_{Q \in \Me} E_Q X = \inf\{x \in \R \sd \exists \ \text{admissible } \phi, \ x + V_T(\phi) \geq X \}.
 \end{equation*}
 Since the infimum is finite,  Theorem 9 in \cite{delbaen1995no} yields that it is indeed a minimum. Without loss of generality, we may chose $x=0$ and obtain the existence of an admissible dynamic trading strategy $\phi$ with $X \leq V_T(\phi)$. As $X \in \overline{A}_b(\G)$ it holds further that $E_P[X \sd \G] \geq 0, \ P$-a.s., which leads us to
 \begin{equation*}
  E_P[V_T(\phi) \sd \G] \geq E_P[X \sd \G] \geq 0 \quad P\text{-a.s.}
 \end{equation*}
 Then, $E_P [V_T(\phi)] \geq E_P [X] > 0$, such that $V_T(\phi) \in \AG$. So the existence of generalized $\G$-arbitrage strategies is equivalent to the existence of $\G$-arbitrage strategies $V_T(\phi)$ in \AG and the claim follows.
\end{proof}

\section{Some classes of profitable strategies}\label{StatArbForDiffProcesses}
In Section \ref{Section3} we saw  conditions and examples of statistical arbitrages in a variety of models. 
Here we are considering several classes of simple statistical arbitrage strategies for several classes of information systems $\G$. While these strategies are easy to apply for general stochastic models 
we investigate them on the Black-Scholes model which will allow for analytic properties of the trading strategies. 
We will see in the following section that similar results can be expected in more general market models.

The Black-Scholes model is, according to  Example \ref{Ex.1}, free of statistical arbitrage, and we show in the following how to construct  dynamic trading strategies allowing statistical $\G$-arbitrage for various choices of $\G$.
 To this end, assume that $S$ is a geometric Brownian motion, i.e.~the unique strong solution of the stochastic differential equation
\begin{equation}\label{Eq.II.3.8}
 dS_t = \mu S_t \dt + \sigma S_t \dB_t, \qquad 0 \le t \le T
\end{equation}
where $B$ is a $P$-Brownian motion and $\sigma>0$. %
In the simulation we will first chose $\mu = 0.1241$, $\sigma = 0.0837$, $S_0 = 2186$ according to estimated drift and volatility from the S\&P 500  (September 2016 to August 2017), and later consider small perturbations.

Motivated by our findings in Section \ref{sec:trinomial}, we begin by embedding binomial trading strategies into the diffusion setting by considering two limits (up / down) and taking actions at the first times these limits are reached. In Section \ref{Sec.ExtBinMod} we will introduce some related follow-the-trend strategies.

\subsection{Embedded binomial trading strategies}\label{Section.II.4.1}

We introduce a recombination of several two-step binomial models
embedded in the continuous-time model as long as the final time $T$ is reached.
As information system we consider the $\sigma$-field $\G$ generated by the stopping times when the final states of each of the binomial model are reached (or the trivial $\sigma$-field otherwise).

As we repeatedly consider embedded binomial models it makes much sense to talk on the outcome of the trading strategy \emph{on average} conditional on the final states of each binomial model, i.e.~by averaging the outcome over many repeated applications of the trading strategy and hence we may apply the concept of \emph{statistical arbitrage} here.

Let $i$ denote the current step of our iteration and consider a multiplicative step size $c>0$. We initialize at time $t_0^0=0$. Otherwise consider the initial time of our next iteration given by the time where finished the last repetition and denote this time by $t_0^i$ and the according level by $s_0^i=S_{t_0^i}$. 
 Then we define the following two stopping times denoting the first and second period of our binomial model by
\begin{equation}\label{StopTime_t1}
 t_1^i = \inf \big\{t \in [t_0^i,T] \sd S_t \in \{ s_0^i (1-c), s_0^i(1+c) \} \big\}
\end{equation}
and
\begin{equation}\label{StopTime_t2}
 t_2^i  = \inf \big\{t \in ( t_1^i,T] \sd S_t \in \{s_0^i(1-2c), s_0^i ,  s_0^i(1+2c)\} \big\},
\end{equation}
with the convention that $\inf \emptyset = T$. This induces a sequence of $\sigma$-fields 
$$ \G^i :=  \sigma(S_{t_2^i}). $$

Since $S$ is continuous, this scheme allows to embed repeated binomial models $S_{t_0^i}, \, S_{t_1^i}, \, S_{t_2^i}$, $i=1,2,\dots$ into continuous time. The considered trading strategy is to execute the statistical arbitrage strategy for binomial models computed in Lemma \ref{Lem.II.3} at the stopping times $t_0^i, t_1^i, \, t_2^i$. 
At $t_2^i$ the position will be cleared and we start the procedure afresh by letting $t_0^{i+1}=t_2^i$. Generally, we assume that the time horizon $T$ is sufficiently large such that the (typically small) levels $s_0^i(1-2c),\dots,s_0^i(1+2c)$ are reached at least once.

\begin{example}
Figure \ref{Fig.II.1} illustrates the embedding of the binomial model: the  boundary $s_0^0(1-c)$ is hit at stopping time $t_1=t_1^0$ and the boundary $s_0^0(1-2c)$ at stopping time $t_2=t_2^0$. 
The trading strategy $\bphi$ from Lemma \ref{Lem.II.3} then implies trading buying (selling)  $\phi_1$ entities of the underlying at time $t=0$ and  $\phi_2^-$ entities at $t = t_1$. At time $t=t_2$ we will close the position 
and start this procedure again with $t^1_0 = t_2$ and with the new starting point $s_0^1 = S_{t_2}$. 
This leads to a recombination of several 2-period binomial models, as illustrated in Figure \ref{Fig.II.2}.
\end{example}

\begin{figure}[t]
  \centering
  \includegraphics[width=9cm]{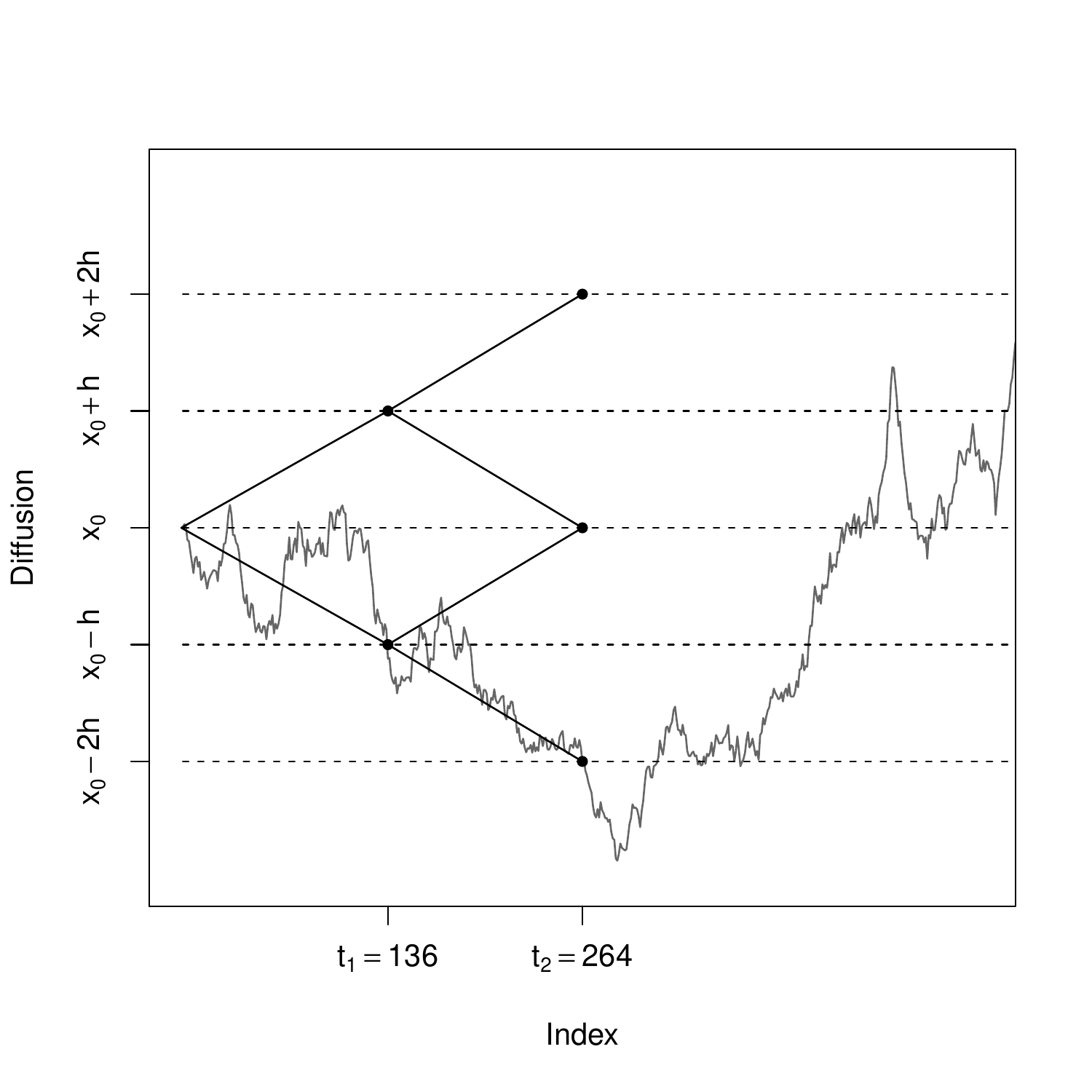}
  \caption{The embedding of a binomial model: at the hitting times $t_1$ and $t_2$ of the diffusion the steps of the embedded binomial model take place. The hitting levels are given by  $ s_0 (1 \pm 0.15)$.
  }\label{Fig.II.1} 
\end{figure}

\begin{figure}[t]
  \centering
  \includegraphics[width=8.5cm]{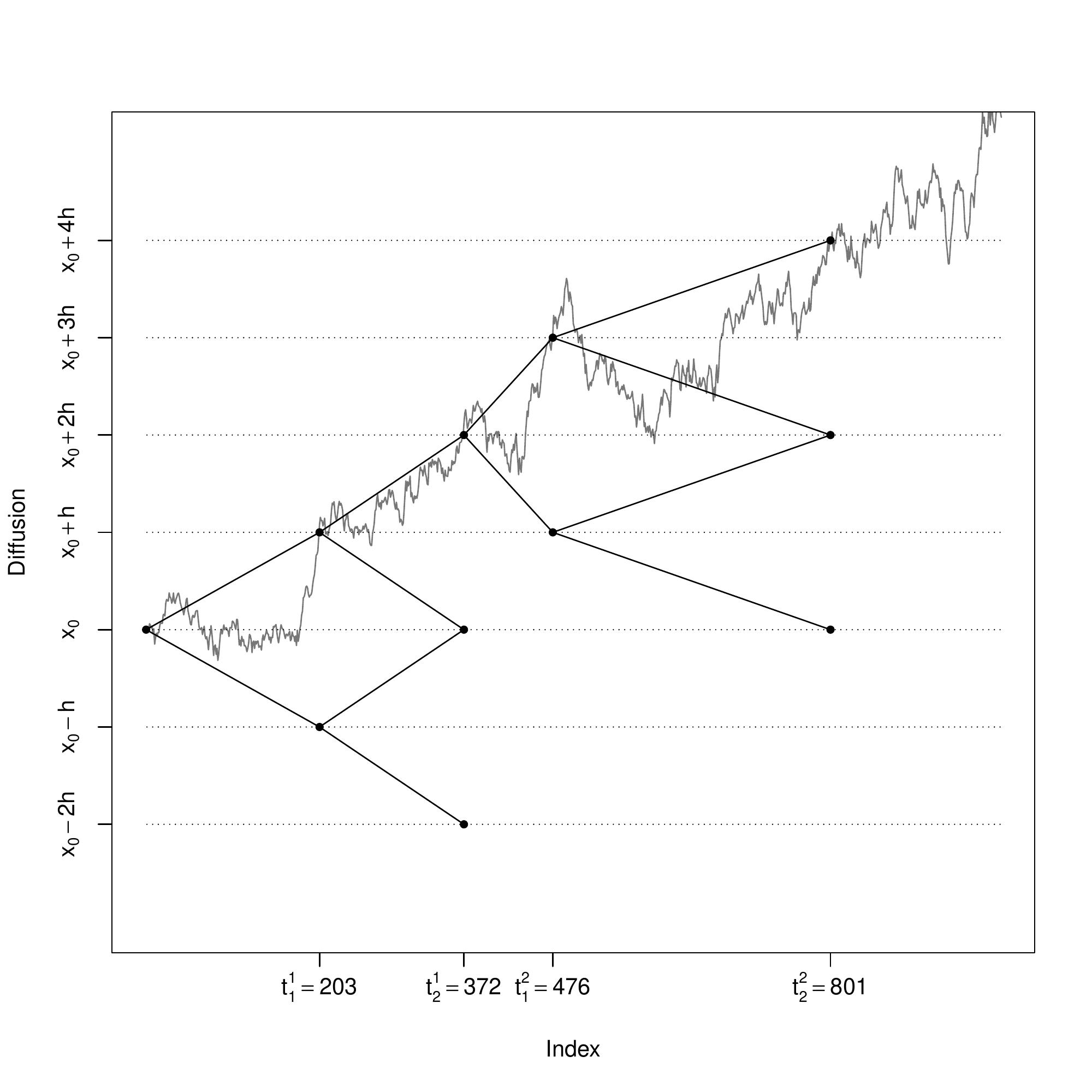}
  \caption{The embedded multi-period binomial trading model with trading points $t_1^1$, $t_2^1$, $t_1^2$ and $t_2^2$. 
  The statistical arbitrage in this case corresponds to repeated trading strategies from Lemma \ref{Lem.II.3}: we buy $\phi_1$ entities at $t_0^1 = 0$, change the position  to $\phi_2^+$ at  $t_1^1$ and equalize the position at $t_2^1$. 
  With the new starting time $t_0^2 = t_2^1$ the strategy will be started again and adjusted at the stopping times $t_1^2$ and $t_2^2$}\label{Fig.II.2}
\end{figure}

The constant $c$ and with it the barriers for the hitting times will be chosen in dependence of $\mu$ and $\sigma$ to ensure that we do not loose the statistical arbitrage opportunity. To be more precise we use 
  $$ c = 0.01  \cdot \frac{\mu}{\sigma}$$ 
which showed a good performance in our simulations.
According to Lemma \ref{Lem.II.2} there is a statistical arbitrage opportunity if $\frac{P(\omega_2)}{P(\omega_3)} \neq \tilde  q$. It is easy to check from Equation \eqref{def:qtilde} that $\tilde q = 1$ in the case considered here.

To guarantee existence of a statistical arbitrage we calculate the path probabilities $P(\omega_2), \, P(\omega_3)$. The first exit time $\tau = \inf \{t \ge 0 \sd S_t \notin (a,b)\}$ from the interval $(a,b)$ satisfies
\begin{equation}\label{Eq.II.4.1.1}
 P(S_\tau = a) = \left(\frac{a}{s_0}\right)^\nu \frac{\big(\frac{b}{s_0}\big)^{|\nu|}-\big(\frac{s_0}{b}\big)^{|\nu|}}{\big(\frac{b}{a}\big)^{|\nu|}-\big(\frac{a}{b}\big)^{|\nu|}}, \hspace{3ex} a<b,
\end{equation}
where   $\nu = \frac{\mu}{\sigma^2} - \frac{1}{2}$,  see  \cite{borodin2012handbook}, formula 3.0.4 in Section 9 of Part II. This in turn yields that
\begin{align}
q &= \frac{P(\omega_2)}{P(\omega_3)} = \frac{P\big(S_{t_1} = s_0(1+c)\big) P\big(S_{t_2} = s_0\big)}{P\big(S_{t_1} = s_0(1-c)\big) P\big(S_{t_2} = s_0\big)}  \notag\\
&= \frac{\bigg(1 - \left(1-c\right)^\nu \frac{(1+c)^{|\nu|}-(1+c)^{-|\nu|}}{\big(\frac{1+c}{1-c}\big)^{|\nu|}-\big(\frac{1-c}{1+c}\big)^{|\nu|}} \bigg) 
 \left(1+c\right)^{-\nu} \frac{\big(\frac{1+2c}{1+c}\big)^{|\nu|}-\big(\frac{1+c}{1+2c}\big)^{|\nu|}}{(1+2c)^{|\nu|}-(1+2c)^{-|\nu|}} }
{\bigg(\left(1-c\right)^\nu \frac{(1+c)^{|\nu|}-(1+c)^{-|\nu|}}{\big(\frac{1+c}{1-c}\big)^{|\nu|}-\big(\frac{1-c}{1+c}\big)^{|\nu|}} \bigg) 
\bigg(1 -  \left(\frac{1-2c}{1-c}\right)^\nu \frac{(1-c)^{-|\nu|}-(1-c)^{|\nu|}}{(1-2c)^{-|\nu|}-(1-2c)^{|\nu|}} \bigg)} .
\end{align}
Clearly, in general $q \neq  1$, such that in these cases statistical arbitrage exists, which we exploit in the following.

From Lemma \ref{Lem.II.3} we obtain with $D = 2 (q  - 2 )  (c\, s_0^i)^3$ that the trading strategy $\bphi=(\phi_1,\phi_2^+,\phi_2^-)$ is given by
\begin{align}
 \phi_1   &= ( 2+q ) (c \, s_0^i)^2 D^{-1} ,\label{Eq.1} \\
 \phi_2^+ &= ( q-4 ) (c \, s_0^i)^2 D^{-1} ,\label{Eq.2} \\
 \phi_2^- &= -  3q (c \, s_0^i)^2   D^{-1}  \label{Eq.3}.
\end{align} 
We call the trading strategy which results by repeated application of $\phi$ at the respective hitting times the \emph{embedded binomial trading strategy}.

\subsection*{Simulation results}
As already mentioned, we simulate a geometric Brownian motion according to Equation \eqref{Eq.II.3.8} with  $\mu = 0.1241$, $\sigma = 0.0837$, $S_0 = 2186$, $T=1$ (year), discretize by 1000 steps and embed the according binomial models repeatedly in this time interval. In this case we have $q=1.00189$ (rounded to five digits) which is not equal to one and therefore $q \neq \tilde{q}$, \ie the embedded binomial strategy in this case is a $\G$-arbitrage strategy. We denote by $N$ the (random) number of binomial models that are necessary for each simulated diffusion to gain either a profit from trading or to reach $T$ and by  $G^i$ the gain or loss of the $i$-th binomial model. Hence either $\sum_{i=1}^N G^i > 0$ or we record a loss at time $N = T$. 

\begin{table}[t]
\begin{center}
{\footnotesize
\begin{tabular}{cccccccc}\toprule 
   gain p.a. & median   & VaR(0.95) & gain/trade  & losses& (mean)  & $\varnothing$ $N$ & max. $N$\\\midrule 
  33.4  &   206   &  5,320  &  8.74  &  0.133  &  -628  &  3.82  &   24  \\ \bottomrule \\
\end{tabular}}
\caption{Simulation results for the \emph{embedded binomial trading strategy}  for 1 mio runs. This example serves as benchmark. Gain p.a.~denotes the overall average \emph{gain} in the time period of one year, $[0,1]$; we also show its median and the associated estimated VaR at level 95\%. 
    \emph{Gain/trade} denotes the average gain per trade, \emph{losses} denotes the fraction of simulations where the outcome of the trading strategy was negative, and we also show the average of the losses titled \emph{mean}. Finally, we also state the average number and maximal number of embedded binomial models.}\label{Tab.II.1}
\end{center}
\end{table}

\begin{figure}[!tb]{%
   \begin{overpic}[height=8.0cm, width = 8.5cm]{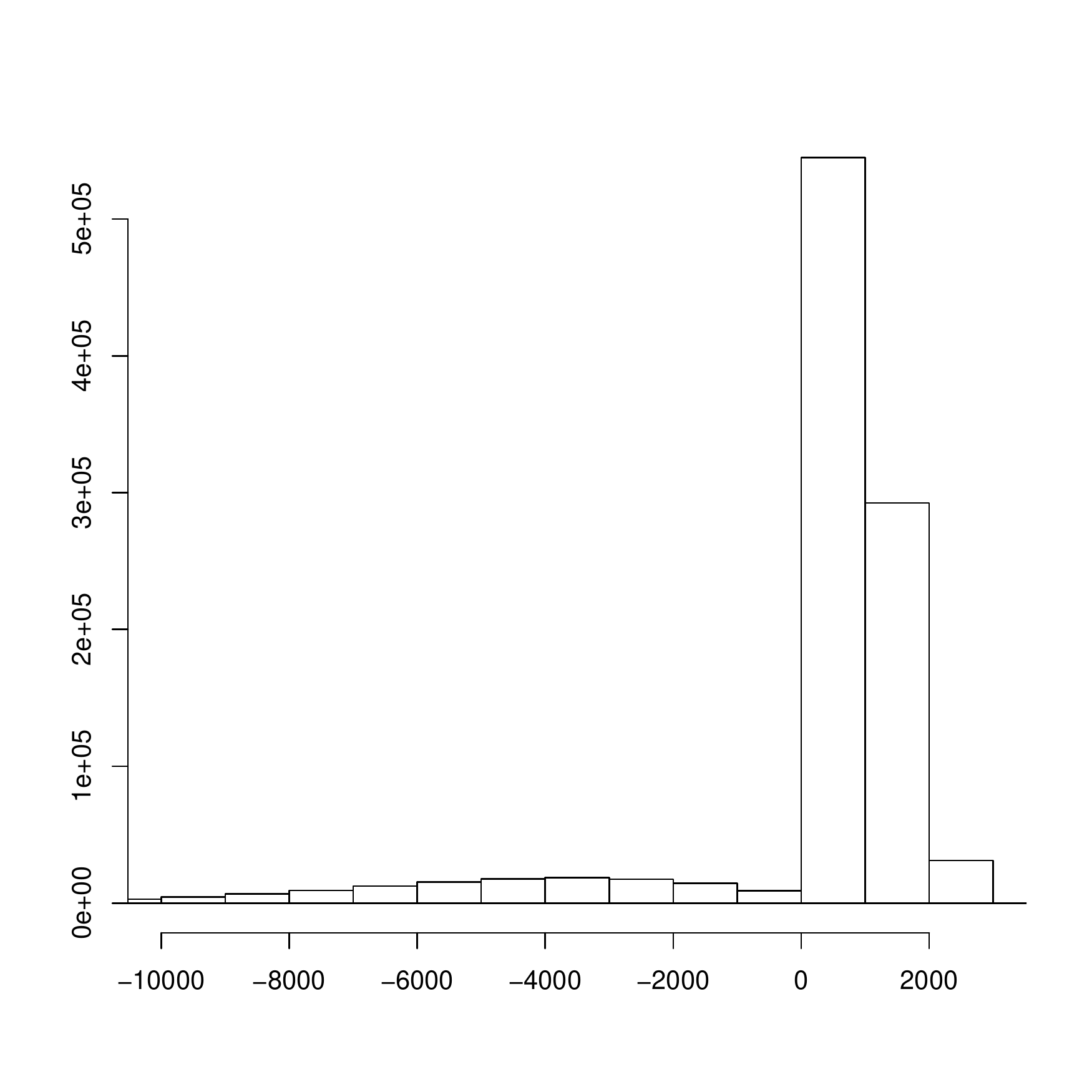}
   \end{overpic}
 }
 \caption{Histogram of the profits and losses from the embedded binomial trading strategy used in Table \ref{Tab.II.1}.}\label{fig:histogram}
\end{figure}

For 1 million runs, we obtain the results presented in Table \ref{Tab.II.1}. 
 For each run we record either a gain or a loss from trading. The average gain per simulation run is shown in column one, its median in column two. The distribution of the P\&L is skewed to the left with potential large losses with small probability which is reflected by a median of 206 in comparison to  an average gain of 33. In column 3 we depict the 95\% Value-at-Risk which is of size 5,320. Column 4 denotes the average gain per trade which is obtained by dividing the average gain by the average number of trades (i.e.~repeated binomial models).  In column 5 we show the (fraction of) \emph{losses}, i.e.~the fraction of simulated  processes exhibiting no gain from trading before reaching the final time $T$, followed by their mean. The average number of trading repeats $\varnothing N$ is followed by  the maximal number of trading repeats over all runs (max $N$).

As becomes clear from  Table \ref{Tab.II.1} we can record an overall profit for many cases.  We have a negative outcome in $13.3$ percent in average of all simulations with an average size of -628. The median of the profits is about 200, with a smaller average of about 30. The risk measured by the Value-at-Risk at 95\% is 5,320 pointing to the fact that the average gain by the statistical arbitrage is (of course) not without risk. For clarification, we plot the associated histogram of the P\&L in Figure \ref{fig:histogram}.

Although the actual amount of the profit depends on many parameters  we can confirm the possibility of statistical arbitrage. 
Besides, we see that on average our multi-period binomial model has a small number of periods and the number of periods does not explode, which is important with a view on  trading costs.

\begin{table}[t]
\begin{center}
{\footnotesize
\begin{tabular}{lrrrrcrrr}\toprule 
  \ \ \ c   & gain pa & median &  VaR${}_{0.95}$  & gain pt & losses  & (mean)  & $\varnothing$ $N$ & (max)\\\midrule 
0.0025 &  8,890  &  48,700    &   -373  &   743  &  0.045  &  -57,900  &    12  &  150  \\ 
0.005  &   465  &  3,810    &  58,400  &  66  &  0.077  &  -6,210  &   7  &   63  \\ 
0.01  &  41  &   206   &  5,250  &  11  &  0.132  &  -621  &  4  &   24  \\ 
0.02  &     9  &   10   &  371  &  5  &  0.185  &  -50  &  2  &    9  \\ 
0.04  &  3  &  2   &  24  &  3  &  0.109  &  -2  &   1  &    4  \\ \bottomrule \\
\end{tabular}}
\caption{
Simulations for the embedded binomial trading strategy with varying boundary levels; gain p.a.~denotes the gain per year, gain p.t.~denotes gain per trade. In the simulations for Table \ref{Tab.II.1} we used $c=0.01 \nicefrac{\mu}{\sigma}$.}\label{Tab.II.2}
\end{center}
\end{table}

\subsubsection*{Varying barrier levels}
The most interesting parameter turns out to be the parameter $c$. It decodes the varying the barrier level and the results may be found in  Table \ref{Tab.II.2}. It turns out that this parameter allows to balance gains and risk very well.

First, the smaller the parameter $c$ is chosen, the higher are the  gains in general. The additional gain does imply an increase of risk: most prominently, the mean of the losses decreases with $c$. On the other side, we observe a decrease in the probability for losses to occur. The Value-at-Risk confirms the increase of risk with decreasing $c$, except for the lowest $c=0.0025$. In this case, the probability of having large losses is below 5\%, such that the Value-at-Risk at level $0.95\%$ does no longer see this risk (while it is of course still present).

A high value of $c$ corresponds intuitively to a larger step sizes, which leads to less trades on average. The largest value of $c$ gives a statistical arbitrage with small gain and smallest risk.

\subsubsection*{The role of drift and volatility}
For the investor it is of interest which drift and which volatility of an asset promises a good profit. To investigate this question we define the fraction 
$$ \eta := \frac{\mu}{\sigma} $$ 
and show simulation results for different values of $\eta$. 
In Table \ref{Tab.II.3} we fix the volatility $\sigma$ and consider varying drift, while in Table \ref{Tab.II.4} we fix the drift $\mu$ and consider varying volatility.

Larger values of $\eta$ point to a high drift relative to volatility situations which we would expect to be very well exploitable. In fact, our simulations show quite the contrary: we observe large gains when  $\eta$ is actually small, while for larger $\eta$ we observe only minor gains. More precisely, for fixed $\sigma$ we obtain decreasing gains for increasing drift, while for fixed $\mu$ we observe increasing gains for increasing volatility. This effect is much more pronounced for the latter case (increasing $\sigma$).
Already from the results with varying step sizes in  Table \ref{Tab.II.2} such an effect was to be expected, as higher values of $\eta$ lead to larger step sizes here and to lower gains. 
Intuitively, larger volatility implies more repetitions and therefore a higher likelihood for the statistical arbitrage to end up with gains. This is also reflected by increasing values of $N$ in Table \ref{Tab.II.4}.

\begin{table}[t]
\begin{center}
{\footnotesize
\begin{tabular}{crrrcccccc}\toprule 
  $\eta$    & gain pa & median &  VaR${}_{0.95}$ & gain pt  & losses & (mean) & $\varnothing$ $N$ & (max)\\\midrule 
0.33  &   211  &  11,600  &  252,000  &  45 &  0.13  &  -29,400  &  5  &   30  \\ 
0.50  &   170  &  4,360  &  94,500  &  36  &  0.13  &  -11,000  &  5  &   30  \\ 
0.75  &   109  &  1,730  &  38,100  &  23  &  0.13  &  \phantom{0}-4,400  &  5  &   30  \\ 
1.00  &    64  &   913  &   20,400  &  14  &  0.12  &  \phantom{0}-2,340  &  5  &   30  \\ 
1.25  &    77  &   561  &   12,400  &    17  &  0.12  &  \hspace{1mm} -1,400  &  5  &   30  \\ 
2.00  &    42  &   197  &   4,430  &  \ 9  &  0.11  &  \phantom{00}\,\,-490  &   4  &   31  \\ 
3.00  &  34  &  81  &       1,680  &  \ 8  &  0.10  &  \phantom{00}\,-182  &  4  &   31  \\ 
 \bottomrule \\[2mm]
\end{tabular}}
\caption{Simulations for the embedded binomial trading strategy with different values of the drift $\mu$ (and hence $\eta$),  fixed $\sigma = 0.1$ and $n = 250,000$ runs; gain p.a.~denotes the gain per year, gain p.t.~denotes gain per trade.}\label{Tab.II.3}
\end{center}
\end{table}

\begin{table}[t]
\begin{center}
{\footnotesize 
\begin{tabular}{rrrrrrrrr}\toprule 
  $\eta$    & gain pa & median  & VaR${}_{0.95}$ & gain pt  & losses & (mean)   & $\varnothing$ $N$ & (max) \\ \midrule 
 0.50  &  74,500  &  222,000  &    -48,400  &  4,340  &  0.036  &  -2,770,000  &  17  &  270  \\ 
0.75  &  6,020  &  59,900  &   480,000  &   582  &  0.056  &  -79,400  &  10  &  120  \\ 
   1.00  &   241  &  4,710  &   80,500  &  37 &  0.090  &  -8,520  &  7  &   51  \\ 
1.25  &  67  &   541  &  12,700  &  16  &  0.124  &  -1,460  &  4  &   28  \\ 
   2.00  &  8  &  6  &  165  &  5 &  0.144  &   -22  &  2  &    9  \\ \bottomrule \\[2mm]
\end{tabular}}
\caption{Simulations for the embedded binomial trading strategy with different values of the volatility (and hence $\eta$), fixed $\mu = 0.1$; gain pa~denotes the gain per year, gain pt~denotes gain per trade.}\label{Tab.II.4}
\end{center}
\end{table}

\subsection{Follow-the-trend strategy}\label{Sec.ExtBinMod} 
As we have seen in the previous section, embedding a binomial model into continuous time 
is not able to exploit a large drift.
This motivates the introduction of a further step into the embedded model in order to
\emph{exploit existing trends} in the underlying. We focus on an \emph{upward trend}, while the strategy is easily adopted to the case for a downward trend.  We consider two-step binomial embedding: first, we specify barriers (up/down) as previously. If we twice observed up movements, we expect an upward  trend and exploit this in a further step. Consequently, here we will consider four stopping times (for iteration $i$): initial time $\tau_0^i$, and stopping times $\tau_1^i$, $\tau_2^i$ as previously and, in addition $\tau_3^i$.  
Most notably, this modelling implies a different choice of the filtration $\G$, see Equation \eqref{eqn:tildeG}.

The associated strategy is to trade in the following way: the first trading occurs as previously at the first time when the barriers $s(1+c)$ or $s(1-c)$ are hit. The next trading takes place when the neighbouring barriers are hit, in the first case $s$ or $s(1+2c)$ and in the second case $s$ or $s(1-2c)$, respectively. If a trend was detected (i.e.~the upper barrier $s(1+2c)$ was hit, as we consider the case of a positive drift), trading continues until a suitable stopping time.

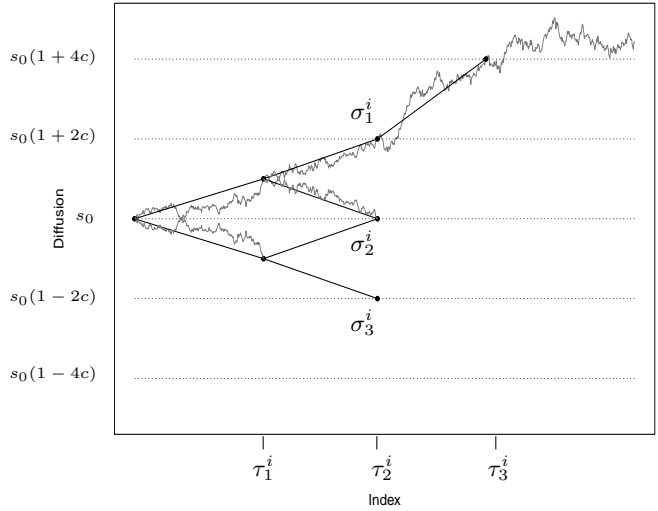
\begin{figure}[!tb]{%
   \begin{overpic}[height=8.0cm, width = 8.5cm]{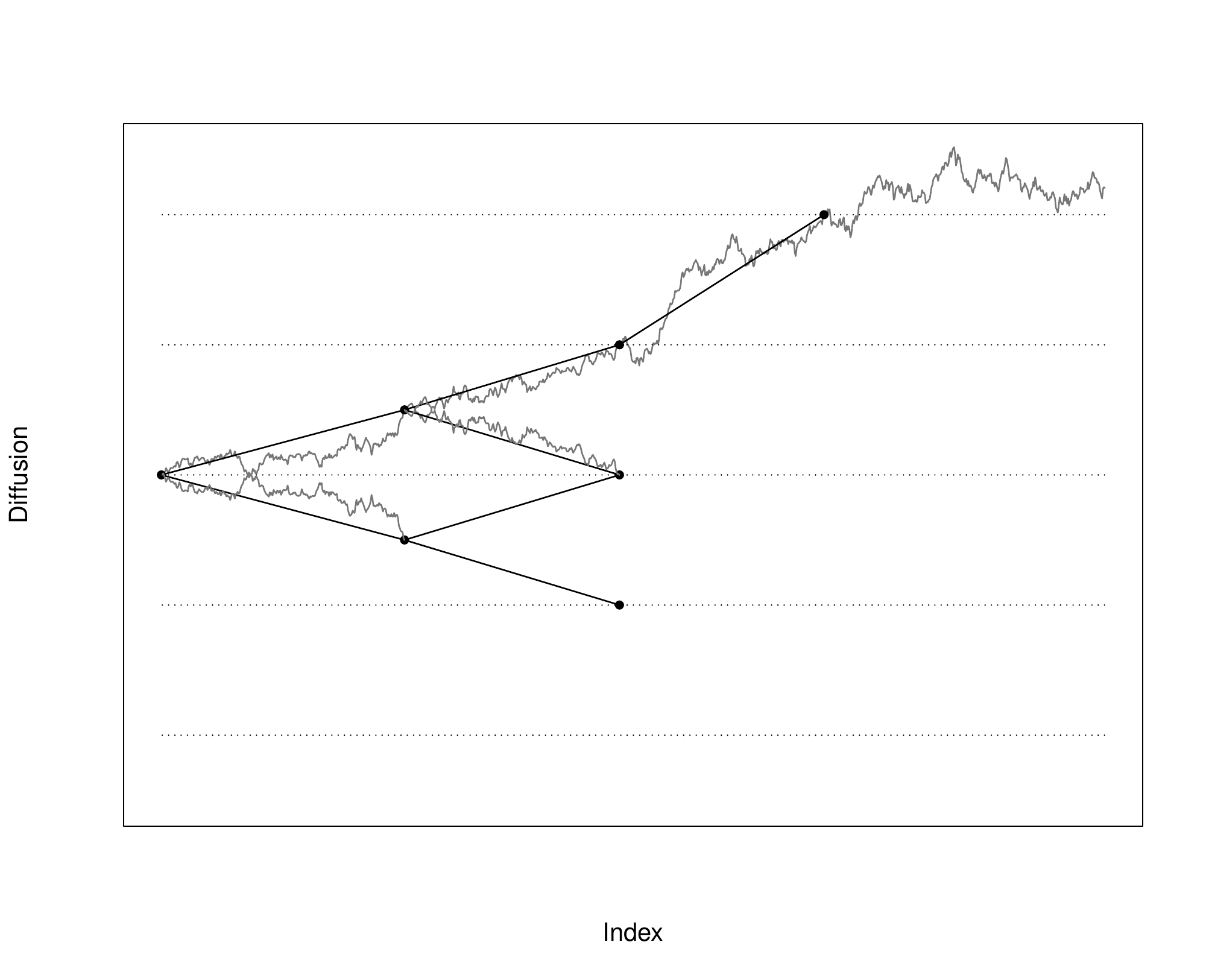}%

    \put(81,30){\line(0,1){5}} %
    \put(77,20){{\small
    $\tau_1^i$
    }}
    \put(114,155){{\small 
    $\sigma_1^i$
    }}
    \put(114,105){{\small 
    $\sigma_2^i$
    }}
    \put(114,75){{\small 
    $\sigma_3^i$
    }}
    \put(124,30){\line(0,1){5}} %
    \put(122,20){{\small
    $\tau_2^i$}}
    \put(169,30){\line(0,1){5}}
    \put(166,20){{\small
    $\tau_3^i$}}
        \put(10,117){{\tiny
    $s_0$}
    }
    \put(-15,86){{\tiny
    $s_0(1-2c)$}
    }
    \put(-15,57){{\tiny
    $s_0(1-4c)$}
    }
    \put(-15,147){{\tiny
    $s_0(1+2c)$}
    }
    \put(-15,177){{\tiny
    $s_0(1+4c)$}
    }
   \end{overpic}
 }
 \caption{Illustration of the stopping times defined in (\ref{eq:PartStopTime1}), (\ref{eq:PartStopTime2}) resp. (\ref{eq:PartStopTime3}). The first stopping takes place when the process reaches either the first upper or lower boundary $s_0^i(1 \pm c)$. Starting from the upper boundary the next stopping takes place if the process increases to the level $s_0^i(1+2c)$, decreases to the level $s_0^i(1-2c)$  or crosses the level $s_0$. In case the process reached the upper level a third stopping occurs at $\tau_3^i$.}\label{fig:StoppingTimes}
\end{figure}

More formally, this leads to the following procedure: let $i$ denote the current step of our iteration. We initialize at time $\tau_0^0=0$. Otherwise consider the initial time of our next iteration given by the the time where we finished the last repetition and denote this time by $\tau_0^i$ and the according level by $s_0^i=S_{\tau_0^i}$. 
Then, using again the property that $S$ is continuous, we define the following successive  stopping times: first, analogously to $t_1^i$ from Equation \eqref{StopTime_t1}, let
\begin{align}\label{eq:PartStopTime1}
 \tau_1^i &= \inf \big\{ t \in (\tau_0^i, T] \sd S_t \geq s_0^i(1+c) \text{ or } S_t \leq s_0^i (1-c)\big\}.
\end{align}
 In the same manner the second stopping occurs if either the upper level is reached, or the mid-level is crossed, or the bottom level is reached. The levels of course differ depending on whether $S_{\tau_1^i}=s_0^i(1+c)$ or $S_{\tau_1^i}=s_0^i(1-c)$. In this regard, we define (for the first case)
 \begin{align*}
     \sigma_1^i & = \inf \big\{ t \in (\tau_1^i, T] \sd  S_t \geq s_0^i(1+2c) \big\}\\
     \sigma_2^i & = \inf \big\{ t \in (\tau_1^i, T] \sd  S_t \leq s_0^i \big\}.
 \end{align*}
 For the second case, we set
 \begin{align*}
     \sigma_3^i & = \inf \big\{ t \in (\tau_1^i, T] \sd  S_t \leq s_0^i(1-2c) \big\}\\
     \sigma_4^i & = \inf \big\{ t \in (\tau_1^i, T] \sd  S_t \geq s_0^i\big\}.
 \end{align*}
 Altogether we obtain that
 \begin{align}\label{eq:PartStopTime2}
  \tau_2^i &=  
  \begin{cases} 
        \sigma_1^i \wedge \sigma_2^i   &  \text{ if } S_{\tau_1^i} = s_0^i(1+c),  \\  
        \sigma_3^i \wedge \sigma_4^i    & \text{otherwise}.
  \end{cases}\end{align}
Finally, we set \begin{align}\label{eq:PartStopTime3}
  \tau_3^i &= 
  \begin{cases}
  \inf \big\{ t \in (\tau_2^i, T] \sd S_t \leq s_0 \text{ or } S_t \geq s_0^i(1+4c) \big\},&\text{if } S_{\tau_2^i} =s_0^i(1+2c),\\
  \tau_2^i,&\text{otherwise.}
  \end{cases}
 \end{align}
 
 Denote by $\tau^{\text{max}}$ the last stopping time of $\tau_3^1,\tau_3^2, \dots$ which lies before $T$. 
 Then the statistical arbitrages traded on the partition of $S_{\tau^{\text{max}}}$ generated by the values $s_0(1 +2kc), \ k=0,1,2,\dots$ which defines the $\G$ on the path space of the diffusion.

 Trading will be executed at times $\tau_1^i$ to $\tau_3^i$ when the process reaches one of the predefined boundaries (or trading time is over). At time $\tau_2^i$ we check if a positive trend persists and trade on this trend. Recall the trading strategy $\bphi=(\phi_1,\phi_2^+,\phi_2^-)$ from Equations \eqref{Eq.1} to \eqref{Eq.3}. First, trading at the first two times  is executed as previously at times $t_0^i,\ t_1^i$, see Lemma \ref{Lem.II.3}: we hold on $[ \tau_0^i,\tau_1^i)$ the fraction  $\phi_1$ shares of $S$. After reaching $s_0^i(1+c)$ ($s_0^i(1-c)$, respectively) at time $\tau_1^i$ the trading strategy changes to holding $\phi_2^+$ ($\phi_2^-$) shares of $S$ until $\tau_2^i$.    The next trading can be split into the following three cases:

\begin{enumerate}[(i)]
 \item $\tau_2^i=\sigma_1^i$: in this case we reached the upper level $s_0^i(1+2c)$ and follow the (upward) trend by holding $\phi_3^{++}$ shares of $S$. This position will be  equalized at $\tau_3^i$ or if the final time is reached.
 \item $\tau_2^i$ equals $\sigma_2^i$ or $\sigma_4^i$: from the state $s_0^i(1+c)$ resp. $s_0^i(1-c)$ we arrived back at $s_0^i$ (or below resp. above). No trend was detected and the embedded binomial trading strategy ends by liquidating the position.
 \item $\tau_2^i$ equals $\sigma_4^i$: again, no (upward) trend was detected and the strategy ends by liquidation the position. 
 \end{enumerate}

\begin{figure}[t]
  \centering
  \begin{tikzpicture}[>=stealth, sloped]
    \matrix (tree) [%
      matrix of nodes,
      minimum size=0.5cm,
      column sep=1.2cm,
      row sep=0.4cm,
    ] 
    {
        &     &   & $S_3(\omega_1)=s^{+++}$\\  
        &     &         &\\
        &         & $S_2(\{\omega_1,\omega_5\})$      &\\
        & $S_1(\{\omega_1,\omega_2,\omega_5\}) $  &             &\\
        $S_0=s $  &         & $S_2(\{\omega_2,\omega_3\}) $   & $S_3(\omega_5)=s^{++-}$\\
        & $S_1(\{\omega_3,\omega_4\}) $ &             &\\
        &         & $S_2(\omega_4)=s^{--} $       &\\
    };
    \draw[-] (tree-5-1)   -- (tree-4-2);%
    \draw[-] (tree-5-1)   -- (tree-6-2);
    
    \draw[-] (tree-4-2)   -- (tree-3-3);%
    \draw[-] (tree-3-3)   -- (tree-1-4);
    \draw[-] (tree-3-3)   -- (tree-5-4);
 
    \draw[-] (tree-4-2)   -- (tree-5-3);
    \draw[-] (tree-6-2)   -- (tree-5-3);%
    \draw[-] (tree-6-2)   -- (tree-7-3);

  \end{tikzpicture}
  \caption{The embedded binomial model for the follow-the-trend strategy with positive drift. The filtration generated by the final states is generated by each $\{\omega_i\}$ for $i=1,4,5$ and $\{\omega_2,\omega_3\}$. We also denote the resulting outcomes by $s=s_0$, $s^+$, $s^-$, \dots and indicate this notation at some places.
  }\label{Fig:Lem53}
\end{figure}
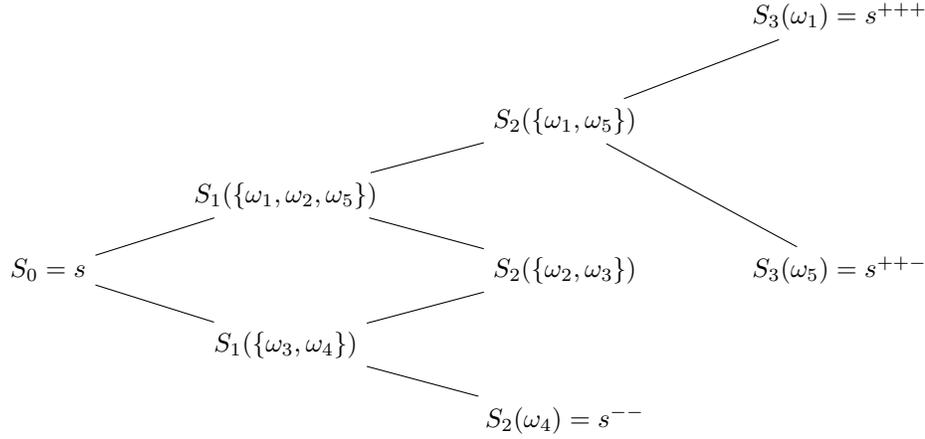
Since Lemma \ref{Lem.II.3} treats a related, but slightly different case we explicitly check in the following that the embedded binomial model indeed allows for {statistical arbitrage}. 

\subsubsection*{The embedded binomial follow-the-trend strategy}We consider $\tilde \Omega=\{\omega_1,\dots,\omega_5\}$ as depicted in Figure \ref{Fig:Lem53}. Let $S_0=s_0 \in \R_{\ge 0}$ and $S_1$ take the two values $s^+$ and $s^-$ such that
$$ S_1(\omega_1)=S_1(\omega_2)=S_1(\omega_5)=s^+, \qquad S_1(\omega_3)=S_1(\omega_4)=s^-. $$
At time $2$ we have the three possibilities $S_2(\omega_1)=S_2(\omega_5)=s^{++}$, $S_2(\omega_2)=S_2(\omega_3)=s^{+-}$ and $S_2(\omega_4)=s^{--}$. In the cases of $\omega_2,\dots,\omega_4$ the model stops. If, however, we saw two up-movements, the model continues and ends up at time $3$ in the states $S_3(\omega_1)=s^{+++}$ or $S_3(\omega_5)=s^{++-}$. 
We assume without loss of generality that $s^+>s_0,$ $s^-<s_0$, and $s^{++}>s^+$, $s^-<s^{+-}<s^+$, and $s^{--}<s^-$ as well as $s^{++-}< s^{++} < s^{+++}$, \ie we consider binomial models as presented in Figure \ref{Fig:Lem53}.

The dynamic trading strategies can be described by
$$ V_3(\phi)=\phi_1 \Delta S_1 + \phi_2 \Delta S_2 + \phi_3 \Delta S_3, $$
with $\phi_1$, $\phi_2^+$, $\phi_2^-$ and $\phi_3^{++}$ being the respective values in the states $\tilde \Omega$, $\{\omega_1,\omega_2,\omega_5\}$, $\{\omega_3,\omega_4\}$ and $\{\omega_1,\omega_5\}$ at times $1,2,$ and $3$, respectively. 
Moreover, we choose 
\begin{align}\label{eqn:tildeG}
   \tilde \G=\sigma(\{\omega_1\},\{\omega_2,\omega_3\},\{\omega_4\},\{\omega_5\}),
\end{align} i.e.~the $\sigma$-field generated by the final states of the embedded binomial model. The following lemma shows that there is always statistical arbitrage in the follow-the-trend strategy if there is statistical arbitrage in the recombining two-period sub-model consisting only of the first two periods.

Denote
  \begin{align}\label{def:gamma} 
    \bgamma =& \frac{1}{D}
       \left( \begin{matrix}
          q \Delta S_2(\omega_2) \Delta S_2(\omega_4) \\
          \Delta S_1(\omega_4) \Delta S_2(\omega_3) - \big( q \Delta S_1(\omega_2) + \Delta S_1(\omega_3)\big) \Delta S_2(\omega_4) \\
          -q \Delta S_2(\omega_2) \Delta S_1(\omega_4)
       \end{matrix}\right)
  \end{align}
with $D$ given in Lemma \ref{Lem.II.3}. The following results shows, that in the follow-the-trend model there is statistical arbitrage, if \eqref{NSArecbinom} holds.

\begin{proposition}\label{Lem.II.4}
 If $\bphi$ is the strategy from Lemma \ref{Lem.II.3}, then for any $\alpha \ge 0$, $\bpsi=(\psi_1,\psi_2^+,\psi_2^-,\psi_3^{++})$ 
with 
$$ \psi_3^{++} = \frac{1-\alpha}{\Delta S_3(\omega_1)-\Delta S_3(\omega_5)} $$
and 
\begin{align*}
      \left( \begin{matrix}
         \psi_1 \\
         \psi_2^+ \\
         \psi_2^-
      \end{matrix}\right)
      &=  \bphi- \Delta S_3(\omega_1) \psi_3^{++} \bgamma
\end{align*}
is a $\tilde \G$-arbitrage strategy, if \eqref{NSArecbinom} holds. 
\end{proposition}
Of course, the possible choice $\alpha=1$ leads to $\psi_3^{++}=0$, such that in this case the statistical arbitrage in the first two periods is exploited and the strategy coincides with that of Lemma \ref{Lem.II.3}. 
\begin{proof}
Following Definition \ref{Def.I.1} the strategy $\bpsi$ is a statistical $\tilde \G$-arbitrage strategy if the following holds
\begin{align} %
 \psi_1 \Delta S_1(\omega_1) + \psi_2^+ \Delta S_2(\omega_1) + \psi_3^{++} \Delta S_3(\omega_1) &\geq 0 \label{Eq.II.4.2.1}\\[1ex]
 \psi_1 \Delta S_1(\omega_4) + \psi_2^- \Delta S_2(\omega_4) &\geq 0 \label{Eq.II.4.2.3}\\[1ex]
 \begin{split}\psi_1 \Delta S_1(\omega_2)P(\omega_2) + \psi_2^+ \Delta S_2(\omega_2) P(\omega_2) & \\
 + \psi_1 \Delta S_1(\omega_3)P(\omega_3) + \psi_2^- \Delta S_2(\omega_3) P(\omega_3) &\geq 0, \end{split} \label{Eq.II.4.2.4} \\[1ex]
 \psi_1 \Delta S_1(\omega_5) + \psi_2^+ \Delta S_2(\omega_5) + \psi_3^{++} \Delta S_3(\omega_5) &\geq 0 \label{Eq.II.4.2.2}
\end{align}
and, in addition, at least one of the inequalities is strict. 

We extend the setting from Lemma \ref{Lem.II.3}. First, we let
  $$
    \tilde A = 
      \begin{pmatrix}
  \Delta S_1(\omega_1)        & \Delta S_2(\omega_1)    & 0 & \Delta S_3(\omega_1) \\
  \Delta S_1(\omega_4)        & 0       & \Delta S_2(\omega_4)& 0 \\
  q \Delta S_1(\omega_2) + \Delta S_1(\omega_3) & q \Delta S_2(\omega_2)  & \Delta S_2(\omega_3)& 0 \\
  \Delta S_1(\omega_5)        & \Delta S_2(\omega_5)    & 0 & \Delta S_3(\omega_5) 
  \end{pmatrix} .$$ 
Then Equations \eqref{Eq.II.4.2.1}--\eqref{Eq.II.4.2.2} are equivalent to $\tilde A \bpsi \ge 0$. 
Note that $S_i(\omega_1)=S_i(\omega_5)$ for $i=1,2$ such that $\tilde A  \bpsi = \tilde \bx$ with $\tilde \bx=(x_1,\dots,x_4)^\top$ reveals 
  $$ \psi_3^{++}= \frac{x_1-x_4}{\Delta S_3(\omega_1)-\Delta S_3(\omega_5)}.$$
As for Lemma \ref{Lem.II.3}, we will consider the case where $\tilde A$ is invertible. Note that the three times three submatrix (upper left) of $\tilde A$ equals the matrix $A$ from Equation \eqref{Eq.II.3.7}. Then, denoting $\bx=(x_1,x_2,x_3)^\top$,
  \begin{align*}
      \left( \begin{matrix}
         \psi_1 \\
         \psi_2^+ \\
         \psi_2^-
      \end{matrix}\right)
      &=  A^{-1} \bx -  A^{-1} \left( \begin{matrix}
    \Delta S_3(\omega_1) \psi_3^{++} \\ 0 \\ 0
    \end{matrix}\right) \\
          &=  A^{-1} \bx - \Delta S_3(\omega_1) \psi_3^{++} \bgamma
  \end{align*}
with vector $\bgamma$ from Equation \eqref{def:gamma}.
Up to now we where free to choose any $\tilde \bx \in \R_{>0}^4$. If we choose, as for Lemma \ref{Lem.II.3}, $\bx=\Ind_3$, then $\bphi=A^{-1}\Ind_3$ is the strategy computed in Lemma \ref{Lem.II.3} and the result follows.
\end{proof}

\subsection*{Simulation results}
We study the performance of the follow-the-trend strategy on basis of various simulations and compare it to the results of the embedded binomial strategies. As previously, we simulate a geometric Brownian motion according to Equation \eqref{Eq.II.3.8} with  $\mu = 0.1241$, $\sigma = 0.0837$, $S_0 = 2186$, $T=1$ (year), discretize by 1000 steps and embed the according models repeatedly in this time interval. In this case, Proposition \ref{Lem.II.4} grants the existence of statistical arbitrage which we will exploit in the following. 

Contrary to the intention of improving the average gain of the follow-the-trend strategy, the simulations show that this goal is not achieved. But, in general, the follow-the-trend strategy leads to a reduction of risk compared to the embedded-binomial trading strategy, visible through the reduced Value-at-Risk  in Tables \ref{Tab.II.5} to \ref{Tab.II.8}. The reduction of the average gain  and its mean  can be explained from  the observations in Section \ref{sec:3.4}: the follow-the-trend-strategy introduces additional scenarios with smaller gains (compare Figure \ref{Fig:Lem53}). This leads to a reduction of the average gain and, at the same time, to a reduction of risk.

\begin{table}[t]
\begin{center}
{\footnotesize
\begin{tabular}{cccccccc}\toprule 
  gain pa & mean & VaR${}_{0.95}$ & gain pt & losses  & (mean) & $\varnothing$ $N$  & (max) \\ \midrule 
  27.8  &   164  &   4,180  &  9.17  &  0.171  &  -554  &  3 &   21 \\
\bottomrule \\[2mm]
\end{tabular} }
\caption{
Simulations for the follow-the-trend strategy for 1 mio runs. In comparison to Table \ref{Tab.II.1} (where the notation is explained) we find slightly smaller gains together with a smaller risk. } \label{Tab.II.5}
\end{center}
\end{table}

The results from Table \ref{Tab.II.6} to \ref{Tab.II.8} show a similar dependence on the choice of the parameters and of the barrier of the follow-the-trend strategy compared to the embedded binomial strategy. In general, we record smaller gains together with smaller risk with one exception: the last line of Table \ref{Tab.II.8} shows that a small $\sigma$ allows the follow-the-trend strategy to exploit the existing (although small) positive trend in the data better. Of course, this comes with a higher risk, which is clearly visible. 

Summarizing, the follow-the-trend strategy  shows (in general) smaller gains together with a smaller risk. The follow-the-trend strategy is, however, able to exploit a positive trend when $\sigma$ is very small.  %

\begin{table}[t]
\begin{center}
{\footnotesize 
\begin{tabular}{rrrrrrrcr}\toprule 
  c \ \ \ \     & gain pa & median & VaR${}_{0.95}$ & gain pt & losses & (mean) & $\varnothing$ $N$ & \hspace{-3mm} (max) \hspace{-3mm}~\\\midrule 
$0.005 \, \nicefrac{\mu}{\sigma} $  &   404  &  3,300    &  51,300  &  71.1  &  0.098  &  -5,590  &  6 &   44  \\ 
$0.01 \, \nicefrac{\mu}{\sigma} $ &  32  &   162   &  4,130  &  10.7  &  0.169  &  -548  &  3  &   18  \\ 
$0.02 \, \nicefrac{\mu}{\sigma} $ &  6  &  8    &  272  &  3.9  &  0.238  &  -45  &  2  &    7  \\ 
$0.04 \, \nicefrac{\mu}{\sigma} $ &   3  &  1   &  23 &   2.6  &  0.122  &  -2  &  1  &    3  \\ \bottomrule \\[2mm]
\end{tabular}}
\caption{Simulations for the follow-the-trend strategy with varying barrier levels $c$. In the simulations for Table \ref{Tab.II.5} we used $c=0.01\, \nicefrac \mu \sigma$.}\label{Tab.II.6}
\end{center}
\end{table}

\begin{table}[t]
  \begin{center}
  {\footnotesize 
  \begin{tabular}{crrrrcrcc}\toprule 
  $\eta$  & gain pa & median  & VaR${}_{0.95}$ & gain pt  & losses & (mean) & $\varnothing$ $N$ & \hspace{-3mm} (max) \hspace{-3mm}~ \\\midrule 
0.33  &   282  &  9,340    &  203,000  &  71  &  0.16  &  -26,100  &  4  &   24  \\ 
0.50  &   122  &  3,500    &  76,200   &  31  &  0.16  &  -9,780   &  4  &   24  \\ 
0.75  &    99  &  1,390    &  30,400   &  26  &  0.16  &  -3,890   &  4  &   22  \\ 
1.00  &    78  &   734     &  16,200   &  20  &  0.15  &  -2,050   &  4  &   23  \\ 
1.25  &    54  &   452     &  9,950    &  15  &  0.15  &  -1,260   &  4  &   23  \\ 
2.00  &    34  &   162     &  3,570    &  10  &  0.14  &  -436     &  3  &   21  \\ 
3.00  &    24  &     66    &  1,390    &  7   &  0.13  &  -165     &  3  &   21  \\ 
 \bottomrule \\[2mm]
\end{tabular}}
    \caption{
    Simulations for the follow-the-trend strategy with varying values of the drift (and hence $\eta=\nicefrac{\mu}{\sigma}$) with fixed $\sigma = 0.1$.}\label{Tab.II.7}
  \end{center}
\end{table}

\begin{table}[t]
  \begin{center}
  {\footnotesize 
  \begin{tabular}{crrrrcrrr}\toprule 
  $\eta$    & gain pa & median &  VaR${}_{0.95}$ & gain pt  & losses & (mean) & $\varnothing$ $N$ & \hspace{-3mm} (max) \hspace{-3mm}~ \\\midrule 
 0.33  &  65,600  &  2,030,000    &  22,700,000  &  6,640  &  0.06  &  -2,770,000  &  10  &  100  \\ 
 0.50  &  2,010  &  40,700        &  586,000    &   284    &  0.09  &  -62,500     &  7  &   58  \\ 
 0.75  &   292      &  3,930      &  69,200     &   60     &  0.12  &  -7,940       &  5  &   34  \\ 
 1.00  &  44     &    732         &  16,400     &  11      &  0.15  &  -2,080      &  4  &   24  \\ 
 1.25  &  27     &   200          &  5,330      &  9       &  0.18  &  -729        &  3  &   17  \\ 
 2.00  &  10     &  15            &  469        &  5       &  0.20  &  -68         &  2  &    9  \\ 
   \bottomrule \\[2mm]
\end{tabular}}

    \caption{
    Simulations for the follow-the-trend strategy with varying values of the volatility $\sigma$ and fixed $\mu = 0.1$.}\label{Tab.II.8}
  \end{center}
\end{table}

\newpage

\subsection{Partition strategies on the final value}\label{Section.GenStatArb}

In this section we study statistical arbitrage with respect to the  information system $\G^{\text{fin}}$ defined  by
\begin{equation} \label{Gdich}
  \{S_T \geq s_0\}=\{\omega_1,\omega_2,\omega_3\}, \hspace{5ex}\text{ and }  \{S_T < s_0\}=\{\omega_4,\omega_5\}.
\end{equation}
This information system corresponds to the two scenarios that the value of the asset increased or decreased at time $T$. The statistical $\G^{\text{fin}}$-arbitrage corresponds to a strategy which yields an average profit in both of these scenarios.

As an example, we continue in the setting of the follow-the-trend model considered in the previous Section \ref{Sec.ExtBinMod}, although other settings are clearly possible. 
Recall that this means we are focusing on an upward trend. We add the assumption that $s^{++-} < s_0$ such that also the third period allows for interesting outcomes (below \emph{or} above $s_0$, compare Figure \ref{Fig:Lem53}). The new information system will lead to a different trading strategy as we detail in the following.

\begin{proposition}
In the follow-the-trend model with $s^{++-} < s_0$ there is  $\G^{\text{fin}}$-arbitrage if 
\begin{align} \label{41}%
 \ \Big( \psi_1 \Delta S_1(\omega_1) + \psi_2^+ \Delta S_2(\omega_1) + \psi_3^{++} \Delta S_3(\omega_1)\Big) & \notag\\[1ex]
+\Big( \psi_1 \Delta S_1(\omega_{3}) + \psi_2^- \Delta S_2(\omega_{3})\Big) \frac{P(\omega_{3})}{P(\omega_1)} & \notag\\[1ex]
 +\Big( \psi_1 \Delta S_1(\omega_2) + \psi_2^+ \Delta S_2(\omega_2)\Big) \frac{P(\omega_2)}{P(\omega_1)} & \ge 0, \\[2mm]
 \  \Big(\psi_1 \Delta S_1(\omega_{4}) + \psi_2^- \Delta S_2(\omega_{4})  \Big)& \notag \\[1ex]
 + \Big(\psi_1 \Delta S_1(\omega_5) + \psi_2^+ \Delta S_2(\omega_5) + \psi_3^{++} \Delta S_3(\omega_5)\Big) \frac{P(\omega_5)}{P(\omega_{4})} &\geq 0 \label{42}
\end{align}
and, in addition, at least one of the inequalities is strict. 
\end{proposition}

The proof is immediate. Note that here there is a lot of freedom in choosing such strategies. Indeed, we will pursue choosing a strategy matching our previous strategies for better comparability.

\begin{example}\label{example:FinValStrat} We consider a special case of \eqref{41}, \eqref{42}: we additionally assume that the first line of Equation \eqref{41} and the first line of Equation \eqref{42} is non-negative. Then, the strategy $\bpsi$ is a $\G^{\text{fin}}$-arbitrage if 
\begin{align}\label{43} 
  \psi_1 \Delta S_1(\omega_1) + \psi_2^+ \Delta S_2(\omega_1) + \psi_3^{++} \Delta S_3(\omega_1) & \ge 0,\\[1ex]
 \psi_1 \Delta S_1(\omega_3) + \psi_2^- \Delta S_2(\omega_{3}) & \notag\\[1ex]
 + \Big(\psi_1 \Delta S_1(\omega_1) + \psi_2^+ \Delta S_2(\omega_2)\Big) \frac{P(\omega_2)}{P(\omega_3)} & \ge 0, \\[2mm]
 \psi_1 \Delta S_1(\omega_3) + \psi_2^- \Delta S_2(\omega_4)  & \ge 0 \\[1ex]
 \psi_1 \Delta S_1(\omega_1) + \psi_2^+ \Delta S_2(\omega_1) + \psi_3^{++} \Delta S_3(\omega_5) &\geq 0,\label{44}
\end{align}
and at least one inequality is strict. Note that we used $\Delta S_1(\omega_3)=\Delta S_1(\omega_4)$,  $\Delta S_1(\omega_1)=\Delta S_1(\omega_2)=\Delta S_1(\omega_5)$ and $\Delta S_2(\omega_1)=\Delta S_2(\omega_5) $ from Section \ref{Sec.ExtBinMod}. This choice is similar to the previously studied  partition strategies and we compute a strategy explicitly. In this regard, define the matrix $A$ by
  $$
     A = 
      \begin{pmatrix}
  \Delta S_1(\omega_1)        & \Delta S_2(\omega_1)    & 0 & \Delta S_3(\omega_1) \\ 
  \Delta S_1(\omega_3)  + r \Delta S_1(\omega_1)       & r \Delta S_2(\omega_2)        & \Delta S_2(\omega_{3})& 0 \\
   \Delta S_1(\omega_3) & 0  & \Delta S_2(\omega_4)& 0 \\
  \Delta S_1(\omega_1)        & \Delta S_2(\omega_1)    & 0 & \Delta S_3(\omega_5) 
  \end{pmatrix}$$
  with $r=\tfrac{P(\omega_2)}{P(\omega_3)}.$
  If $A$ is invertible, for any $\alpha \ge 0$, the strategy $\bpsi$ given by
  $$ \psi_3^{++} = \frac{1-\alpha}{\Delta S_3(\omega_1)-\Delta S_3(\omega_5)} $$
  and 
  \begin{align*}
      \left( \begin{matrix}
         \psi_1 \\
         \psi_2^+ \\
         \psi_2^-
      \end{matrix}\right)
      &= \bphi - \Delta S_3(\omega_1) \psi_3^{++} \bgamma
\end{align*}
is a $\G^{\text{fin}}$-arbitrage. 
Here, 
$\bphi = \tfrac 1 D  (\xi^1, \xi^2, \xi^3)$ with
 \small{\begin{align*}
  \xi^1 &= 
  \Big( r\Delta S_2(\omega_2) - \Delta S_2(\omega_1)  \Big) \Delta S_2(\omega_4) +  \Delta S_2(\omega_1) \Delta S_2(\omega_3) ,\\
  \xi^2 &= \Big( \Delta S_1(\omega_3) - \Delta S_1(\omega_1) \Big) \Delta S_2(\omega_3) 
         + \Big( \Delta S_1(\omega_1) - \Delta S_1(\omega_3) - r \Delta S_1(\omega_1)  \Big) \Delta S_2(\omega_4),\\
  \xi^3 &= r \Delta S_1(\omega_1) \Big( \Delta S_2(\omega_2)-\Delta S_2(\omega_1) \Big) - r \Delta S_2(\omega_2) \Delta S_1(\omega_3),
 \intertext{and}
    D   & = \Big(r \Delta S_1(\omega_1) \Delta S_2(\omega_2) - \big(\Delta S_1(\omega_3)    +r \Delta S_1(\omega_2) \big) \Delta S_2(\omega_1)  \Big) \Delta S_2(\omega_4) \\
        &+ \Delta S_1(\omega_3) \Delta S_2(\omega_1) \Delta S_2(\omega_{3}),  
 \end{align*}} 
 computed analogously to Lemma \ref{Lem.II.3}. In addition,
\begin{align*}
\bgamma = \frac{1}{D}
       \left( \begin{matrix}
          r \Delta S_2(\omega_2) \Delta S_2(\omega_4) \\
          \Delta S_1(\omega_3) \Delta S_2(\omega_3) - \big( r \Delta S_1(\omega_1) + \Delta S_1(\omega_3)\big) \Delta S_2(\omega_4) \\
          -r \Delta S_2(\omega_2) \Delta S_1(\omega_3)
       \end{matrix}\right), %
\end{align*}
and the computation of the strategy is finished. \hfill $\diamond$ 
\end{example}

\begin{remark}\label{Remark.Dich.Strat}
Under the same assumptions as in the previous example we aim to find a $\G^{\text{fin}}$-arbitrage strategy fulfilling equations (\ref{43}) - (\ref{44}).
In that case the strategy $(\Phi, \psi^{++})$ with $\Phi = (\xi^1, \xi^2, \xi^3)$ as in Lemma \ref{Lem.II.3} and 
\begin{equation*}
 -\frac{1}{\Delta S_3(\omega_1)} \leq \psi^{++} \leq -\frac{1}{\Delta S_3(\omega_5)}
\end{equation*}
is a $\G^{\text{fin}}$-arbitrage strategy. To see this remind that 
\begin{align*} 
  \xi^1 \Delta S_1(\omega_1) + \xi^2 \Delta S_2(\omega_1) & \ge 0,\\[1ex]
 \xi^1 \Delta S_1(\omega_2) + \xi^2 \Delta S_2(\omega_2) & \notag\\[1ex]
 + \Big(\xi^1 \Delta S_1(\omega_3) + \xi^3 \Delta S_2(\omega_3)\Big) \frac{P(\omega_3)}{P(\omega_2)} & \ge 0, \\[2mm]
 \xi^1 \Delta S_1(\omega_4) + \xi^3 \Delta S_2(\omega_4)  & \ge 0 \\[1ex]
 \xi^1 \Delta S_1(\omega_5) + \xi^2 \Delta S_2(\omega_5) &\geq 0,
\end{align*}
where $\xi^1 \Delta S_1(\omega_1) + \xi^2 \Delta S_2(\omega_1) = \xi^1 \Delta S_1(\omega_5) + \xi^2 \Delta S_2(\omega_5)$. We are looking for $\psi^{++}$ with 
\begin{align*}
 B + \psi^{++} \Delta S_3(\omega_1) \geq 0,\\
 B + \psi^{++} \Delta S_3(\omega_5) \geq 0,
\end{align*}
where $B := \xi^1 \Delta S_1(\omega_1) + \xi^2 \Delta S_2(\omega_1)$. This results in
\begin{align*}
 \frac{B}{\Delta S_3(\omega_1)} \geq - \psi^{++},\\
 \frac{B}{\Delta S_3(\omega_5)} \leq - \psi^{++}.
\end{align*}
Note that $B \geq 0$, as $\Phi$ is a statistical arbitrage strategy. Besides that we have $\Delta S_3 (\omega_1) > 0$ and $\Delta S_3 (\omega_5) < 0$ and we therefore obtain
\begin{equation*}
 -\frac{B}{\Delta S_3(\omega_1)} \leq \psi^{++} \leq -\frac{B}{\Delta S_3(\omega_5)}.
\end{equation*}
As $B$ was set equal to 1 in Lemma \ref{Lem.II.3} we gain in this setting the special condition
\begin{equation*}
 -\frac{1}{\Delta S_3(\omega_1)} \leq \psi^{++} \leq -\frac{1}{\Delta S_3(\omega_5)},
\end{equation*}
but of course strategies can be derived for any $B \geq 0$. 
\end{remark}

\subsection*{Simulation results}
Again, we study the performance of the strategy, this time the strategy derived in Example \ref{example:FinValStrat} with a partition (above/below) on the final value of the stock. We perform various simulations. As previously, we simulate a geometric Brownian motion according to Equation \eqref{Eq.II.3.8} with  $\mu = 0.1241$, $\sigma = 0.0837$, $S_0 = 2186$, $T=1$ (year), discretize by 1000 steps and embed the according models repeatedly in this time interval. The properties for existence of a v in this setting are confirmed numerically.

As pointed out before, the statistical arbitrages are with respect to different information fields. %
By our variant of $\G^{\text{fin}}$-arbitrage chosen in Example \ref{example:FinValStrat} we indeed find very similar results to the follow-the-trend strategy as one can see in Table \ref{Tab.II.9} to \ref{Tab.II.12}.

\begin{table}[t]
\begin{center}
{\footnotesize 
\begin{tabular}{ccccccccc}\toprule 
   gain pa & median & VaR${}_{0.95}$  & gain pt & losses & (mean) & $\varnothing$ $N$ &  \hspace{-3mm} (max) \hspace{-3mm}~ \\\midrule 
  28.6  &   167    &  4,290  &  8.76  &  0.158  &  -544  &  3  &   20  \\ 
\bottomrule \\[2mm]
\end{tabular}}
\caption{Statistical $\G^{\text{fin}}$-arbitrage trading strategy simulation results for 1 mio simulations with $c^i = 0.01 \, \eta \, S_{\sigma_0^i}$. In comparison to Table \ref{Tab.II.1} (the embedded binomial strategy) we find slightly smaller gains together with smaller losses, while the gains are larger than in Table \ref{Tab.II.5} (the follow-the-trend strategy).  }\label{Tab.II.9}
\end{center}
\end{table}

\begin{table}[t] 
\begin{center}
{\footnotesize 
\begin{tabular}{rrrrrrrcc}\toprule 
  c  \ \ \ \  & gain pa & median &  VaR${}_{0.95}$ & gain pt  & losses &(mean)  & $\varnothing$ $N$ &  \hspace{-3mm} (max) \hspace{-3mm}~ \\\midrule 
$0.005 \, \nicefrac{\mu}{\sigma} $  &   356  &  3,280   &  51,500  &  58 &  0.09  &  -5,510  &  6  &   49  \\ 
$0.01 \, \nicefrac{\mu}{\sigma} $   &    28  &   166    &  4,290   &  9  &  0.15  &  -543  &  3  &   19  \\ 
$0.02 \, \nicefrac{\mu}{\sigma} $   &   6    &       8  &     288  &  4  &  0.22  &  -44  &  2  &   \  8  \\ 
 $0.04 \, \nicefrac{\mu}{\sigma} $  &   3    &       1  &     22   &  3  &  0.12  &  -2  &  1  &   \  4  \\ 
\bottomrule \\[2mm]
\end{tabular}}
\caption{Simulation results for the 
statistical $\G^{\text{fin}}$-arbitrage trading strategy with varying boundaries of the embedded binomial model. Gain p.t. is gain per trade and $\bar N$ equals the maximal $N$ in the simulations.}\label{Tab.II.10}
\end{center}
\end{table}

\begin{table}[t] 
  \begin{center}{\footnotesize 
  \begin{tabular}{crcrcccccc}\toprule 
  $\eta$    & gain pa & median &  VaR${}_{0.95}$ & gain pt  & losses &(mean)  & $\varnothing$ $N$ & \hspace{-3mm} (max) \hspace{-3mm}~ \\\midrule 
0.33  &   192  &  9,600  &    207,000  &  45.2  &  0.15  &  \hspace{-2mm}-25,700  &  4  &   26  \\ 
0.50  &   112  &  3,560  &    77,700  &  26.7  &  0.15  &  -9,600    &  4  &   25  \\ 
0.75  &    97  &  1,430  &    31,200  &  23.7  &  0.14  &  -3,830    &  4  &   26  \\ 
1.00  &    73  &   \phantom{0,}751   &    16,600  &  18.3  &  0.14  &  -2,020  &  4  &   26  \\ 
1.25  &    55  &   \phantom{0,}458   &    10,100  &  13.9  &  0.14  &  -1,230  &  4  &   24  \\ 
2.00  &    34  &   \phantom{0,}163   &   3,600  &  9.15  &  0.13  &  \phantom{0}\,-428  &  3 &   25  \\ 
3.00  &    24  &   \phantom{0,0}67   &    1,410  &  6.82  &  0.12  &  \phantom{0}\,-162  &  3  &   24  \\ 
 \bottomrule \\[2mm]
\end{tabular}}

    \caption{Statistical $\G^{\text{fin}}$-arbitrage trading strategy for varying $\mu$ but with fixed  $\sigma = 0.01$. Gain p.t. is gain per trade and $\bar N$ equals the maximal $N$ in the simulations.}\label{Tab.II.11}
  \end{center}
\end{table}

\begin{table}[t]
\begin{center}
{\footnotesize 
\begin{tabular}{crrrccrccc}\toprule 
  $\eta$    & gain pa & median  & VaR${}_{0.95}$ &  gain pt & losses & (losses) & $\varnothing$ $N$ &  \hspace{-3mm} (max) \hspace{-3mm}~ \\\midrule 
0.75  &   203  &  3,890    &  69,800  &  38  &  0.11    &  -7,810  &  5  &   37  \\ 
1.00  &  71    &   752     &  16,600  &  18  &  0.14    &  -2,020  &  4  &   25  \\ 
1.25  &  28    &   205     &  5,500   &  \ 9   &  0.17  &  -715  &  3  &   18  \\ 
2.00  &  10    &  15       &  494     &  \ 5   &  0.19  &  -67  &  2  &   11  \\ 
3.00  &  4  &    3         &  51      &  \ 3   &  0.09  &  -5  &   1  &   \ 6  \\ 
\bottomrule \\[2mm]
\end{tabular}}

\caption{Statistical $\G^{\text{fin}}$-arbitrage trading strategy for varying $\sigma$ but fixed  $\mu = 0.1$.}\label{Tab.II.12}
\end{center}
\end{table}

\pagebreak

\subsection{Summary on the different strategies}
The previous results confirm statistical $\G$-arbitrage for all three introduced  strategies with respect to the corresponding choices of $\G$. 
Although we observe similar patterns through all strategies like higher gains for smaller boundaries or an decreasing average profit for increasing $\eta$ there are significant differences between the strategies:
\begin{enumerate}[(i)]
 \item  the \emph{average profit} achieved is best for the embedded binomial strategy. 
 \item 	The follow-the-trend strategy and the $\G^{\text{fin}}$-arbitrage strategy show similar behaviour: while showing smaller gains on average, these two strategies have smaller risk. 
\end{enumerate}

\section{Application to market data}

In this section we apply the previously studied approaches to real stock data. 
It is quite remarkable that the positive impression from the simulated data persists on market data.
We study  data from the Kellogg Company and from Deutsche Bank and study the performance of the $\G^{\text{fin}}$-arbitrage
from Chapter \ref{Section.GenStatArb}.

Before we can start with that we have to do some preparations. As we determined the strategies above assuming a positive drift we have to calculate the corresponding strategy for negative drift at first. This is because we will determine the drift in the following examples using real market data and in this case of course there will be both, sections with positive and negative drift as well.\\
We consider $\tilde \Omega=\{\omega_1,\dots,\omega_5\}$ as depicted in Figure \ref{Fig:FinValNegDrift}. Let $S_0=s_0 \in \R_{\ge 0}$ and $S_1$ take the two values $s^+$ and $s^-$ such that
$$ S_1(\omega_1)=S_1(\omega_2)=s^+, \qquad S_1(\omega_3)=S_1(\omega_4)=S_1(\omega_5)=s^-. $$
At time $2$ we have the three possibilities $S_2(\omega_1)=s^{++}$, $S_2(\omega_2)=S_2(\omega_3)=s^{+-}$ and $S_2(\omega_4)=S_2(\omega_5)=s^{--}$. In the cases of $\omega_1,\dots,\omega_3$ the model stops. If, however, we saw two down-movements, the model continues and ends up at time $3$ in the states $S_3(\omega_4)=s^{---}$ or $S_3(\omega_5)=s^{--+}$. 
We assume without loss of generality that $s^+>s_0,$ $s^-<s_0$, and $s^{++}>s^+$, $s^-<s^{+-}<s^+$, and $s^{--}<s^-$ as well as $s^{---}< s^{--} < s^{--+}$, \ie we consider binomial models as presented in Figure \ref{Fig:FinValNegDrift}.
 \begin{figure}[t]
  \centering
  \begin{tikzpicture}[>=stealth, sloped]
    \matrix (tree) [%
      matrix of nodes,
      minimum size=0.5cm,
      column sep=1.2cm,
      row sep=0.4cm,
    ] 
    {
        &                                     & $S_2(\omega_1)=s^{++}$    &\\
        & $S_1(\{\omega_1,\omega_2\}) $             &                     &\\
    $S_0=s $  &                                     & $S_2(\{\omega_2,\omega_3\})$  & $S_3(\omega_5)=s^{--+}$\\
        & $S_1(\{\omega_3,\omega_4,\omega_5\}) $    &                     &\\
        &                                     & $S_2(\{\omega_4,\omega_5\})$    &\\
        &                                     &                     &\\
        &                                     &                               & $S_3(\omega_4)=s^{---}$\\ 
    };
    \draw[-] (tree-3-1)   -- (tree-2-2);%
    \draw[-] (tree-3-1)   -- (tree-4-2);
    
    \draw[-] (tree-2-2)   -- (tree-1-3);%
    \draw[-] (tree-2-2)   -- (tree-3-3);%
 
    \draw[-] (tree-4-2)   -- (tree-3-3);
    \draw[-] (tree-4-2)   -- (tree-5-3);%
   
    \draw[-] (tree-5-3)   -- (tree-3-4);
    \draw[-] (tree-5-3)   -- (tree-7-4);

  \end{tikzpicture}
  \caption{The embedded binomial model for the follow-the-trend strategy with negative drift. The filtration generated by the final states is generated by each $\{\omega_i\}$ for $i=1,4,5$ and $\{\omega_2,\omega_3\}$. We also denote the resulting outcomes by $s=s_0$, $s^+$, $s^-$, \dots and indicate this notation at some places.
  }\label{Fig:FinValNegDrift}
\end{figure}
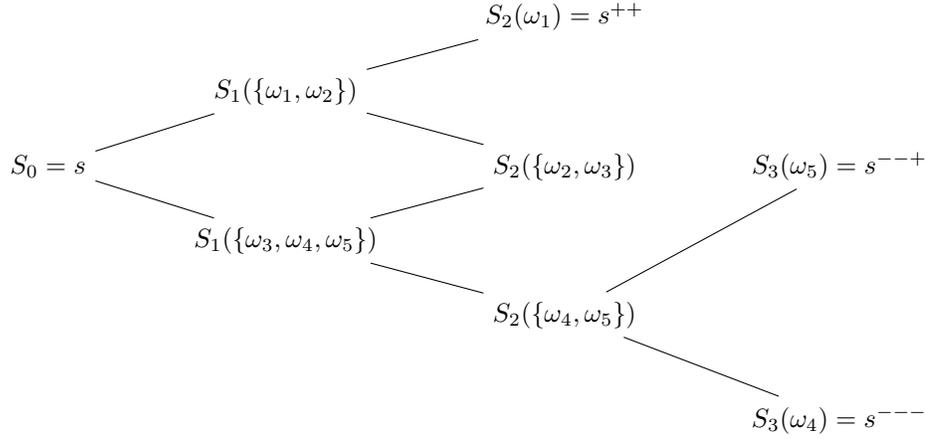

We have a look at statistical arbitrage with respect to the  information system $\G^{\text{fin}}$ defined  by
\begin{equation} \label{Gdich}
  \{S_T > s_0\}=\{\omega_1,\omega_5\}, \hspace{5ex}\text{ and }  \{S_T \leq s_0\}=\{\omega_2,\omega_3,\omega_4\}.
\end{equation}
Analogously to the case with positive drift we add the assumption that $s^{--+} > s_0$. This will lead to a different trading strategy as we detail in the following. 

\begin{proposition}
In the follow-the-trend model with $s^{--+} > s_0$ there is  $\G^{\text{fin}}$-arbitrage if 
\begin{align}\label{45}
 \ \Big( \psi_1 \Delta S_1(\omega_1) + \psi_2^+ \Delta S_2(\omega_1) \Big) P(\omega_1) & \notag\\[1ex]
+\Big( \psi_1 \Delta S_1(\omega_{5}) + \psi_2^- \Delta S_2(\omega_{5}) + \psi_3^{--} \Delta S_3(\omega_{5}) \Big) P(\omega_{5}) & \ge 0, \\[2mm]
 \  \Big(\psi_1 \Delta S_1(\omega_{2}) + \psi_2^+ \Delta S_2(\omega_{2})  \Big) P(\omega_2) & \notag \\[1ex]
 + \Big(\psi_1 \Delta S_1(\omega_{3}) + \psi_2^- \Delta S_2(\omega_{3})  \Big) P(\omega_3) & \notag \\[1ex]
 + \Big(\psi_1 \Delta S_1(\omega_4) + \psi_2^- \Delta S_2(\omega_4) + \psi_3^{--} \Delta S_3(\omega_4) \Big) P(\omega_4) &\geq 0 \label{46}
\end{align}
and, in addition, at least one of the inequalities is strict. 
\end{proposition}

\begin{example}\label{example:FinValStratNegDrift} We consider a special case of \eqref{45}, \eqref{46}: we additionally assume that the first line of Equation \eqref{45} and the last line of Equation \eqref{46} is non-negative. Then, the strategy $\bpsi = (\psi, \psi^+, \psi^-, \psi^{--})$ is a $\G^{\text{fin}}$-arbitrage if 
\begin{align*} 
  \psi_1 \Delta S_1(\omega_1) + \psi_2^+ \Delta S_2(\omega_1) & \ge 0,\\[1ex]
 \Big(\psi_1 \Delta S_1(\omega_{2}) + \psi_2^+ \Delta S_2(\omega_{2})\Big) \frac{P(\omega_2)}{P(\omega_3)}  & \notag\\[1ex]
 + \psi_1 \Delta S_1(\omega_{3}) + \psi_2^- \Delta S_2(\omega_{3}) & \ge 0, \\[2mm]
 \psi_1 \Delta S_1(\omega_4) + \psi_2^- \Delta S_2(\omega_4) + \psi_3^{--} \Delta S_3(\omega_4)  & \ge 0 \\[1ex]
 \psi_1 \Delta S_1(\omega_{5}) + \psi_2^- \Delta S_2(\omega_{5}) + \psi_3^{--} \Delta S_3(\omega_{5}) &\geq 0,
\end{align*}
and at least one inequality is strict. In this regard, define the matrix $A$ by
  $$
     A = 
      \begin{pmatrix}
  \Delta S_1(\omega_1)        & \Delta S_2(\omega_1)    & 0 & 0 \\ 
  \Delta S_1(\omega_3)  + r \Delta S_1(\omega_2)       & r \Delta S_2(\omega_2)        & \Delta S_2(\omega_{3})& 0 \\
   \Delta S_1(\omega_4) & 0  & \Delta S_2(\omega_4)& \Delta S_3(\omega_4) \\
  \Delta S_1(\omega_5)        & 0    & \Delta S_2(\omega_5) & \Delta S_3(\omega_5) 
  \end{pmatrix}$$
  with $r=\tfrac{P(\omega_2)}{P(\omega_3)}.$
  If $A$ is invertible, for any $\alpha \ge 0$, the strategy $\bpsi$ given by
  $$ \psi_3^{--} = \frac{1-\alpha}{\Delta S_3(\omega_4)-\Delta S_3(\omega_5)} $$
  and 
  \begin{align*}
      \left( \begin{matrix}
         \psi_1 \\
         \psi_2^+ \\
         \psi_2^-
      \end{matrix}\right)
      &= \bphi - \Delta S_3(\omega_4) \psi_3^{--} \bgamma
\end{align*}
is a $\G^{\text{fin}}$-arbitrage. 
Here $\bphi$ is the strategy from Lemma \ref{Lem.II.3} and
\begin{align*}
\bgamma = \frac{1}{D}
       \left( \begin{matrix}
          \Delta S_2(\omega_1) \Delta S_2(\omega_3) \\
          -\Delta S_1(\omega_1) \Delta S_2(\omega_3) \\
          - \Delta S_2(\omega_1) \big( r \Delta S_1(\omega_2) + \Delta S_1(\omega_3)\big) + r \Delta S_1(\omega_1) \Delta S_2(\omega_2)
       \end{matrix}\right),
\end{align*}
with
\begin{align*}
    D   & = \Big(r \Delta S_1(\omega_1) \Delta S_2(\omega_2) - \big(\Delta S_1(\omega_3)    +r \Delta S_1(\omega_2) \big) \Delta S_2(\omega_1)  \Big) \Delta S_2(\omega_4) \\
        &+ \Delta S_1(\omega_3) \Delta S_2(\omega_1) \Delta S_2(\omega_{3}).
\end{align*}
\end{example}
The approach now is to simulate the trading with a dynamic strategy, i.\,e. whenever the data leads to a positive drift we will use the strategy from Example \ref{example:FinValStrat} while for  a negative drift we will use the strategy described above.

\begin{example}[Kellogg Company]
In Figure \ref{Fig.II.10} we depict historical stock prices of the Kellogg Company from January 1, 2000 to December 31, 2017. 
\begin{figure}[t]
  \centering
  \includegraphics[height=7cm, width = 10.0cm]{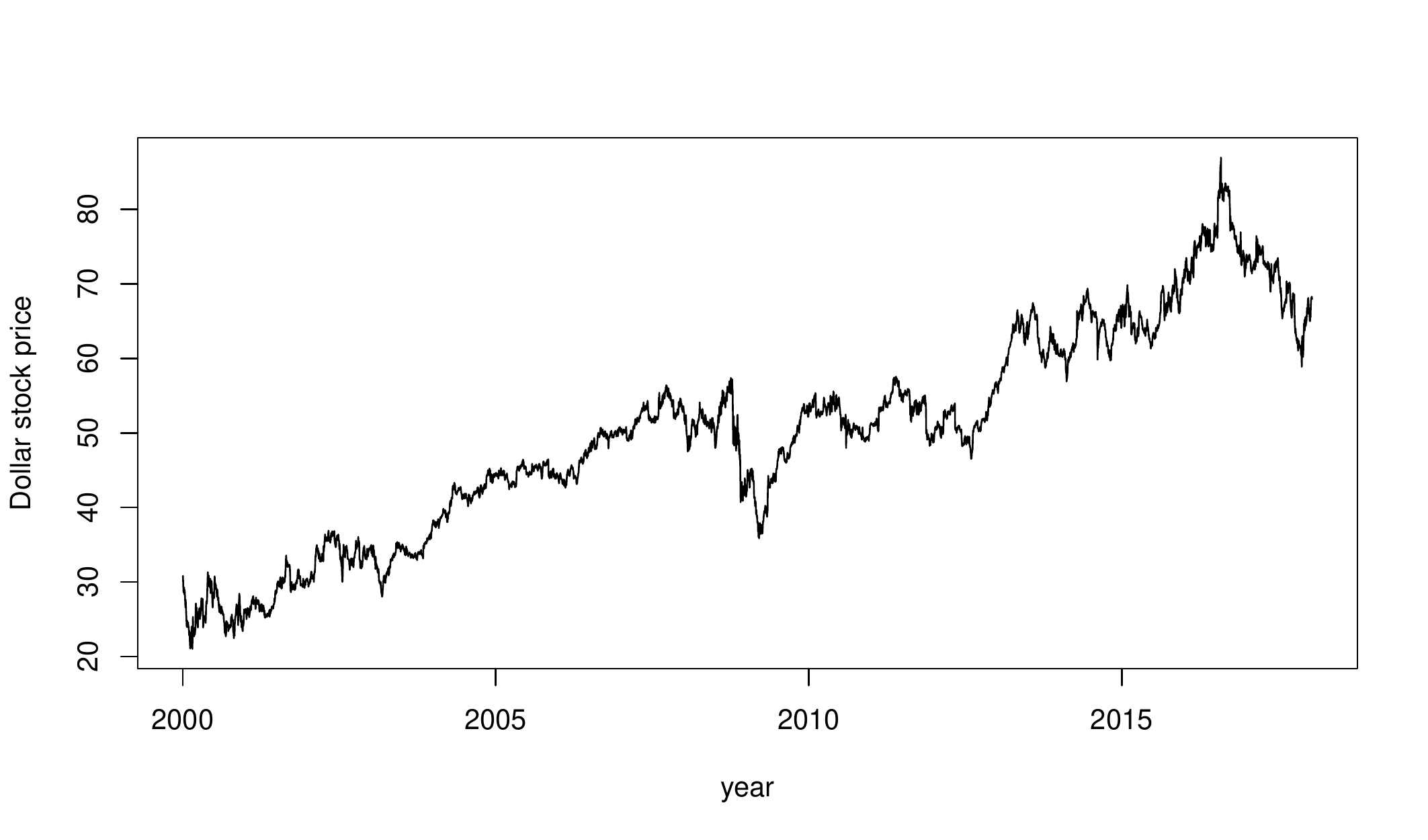}
  \caption{Daily closing prices of the shares of the Kellogg Company during January 1, 2000 and December 31, 2017. Prices are presented in US-Dollar.}\label{Fig.II.10}
\end{figure}
Trading strategies are used by implementing the strategies from Example \ref{example:FinValStrat} and \ref{example:FinValStratNegDrift} where the parameters of the geometric Brownian motion are estimated by the maximum-likelihood estimates from three years directly before the trading period (which is a sliding-window approach with a window length of 3 years).
Table \ref{Tab.II.13} shows the achieved gains for different boundary values. The gains are normalized to one traded asset to improve comparability. 
\begin{table}[t]
\begin{center}
\begin{tabular}{crr}\toprule
 boundary     & GPTA: Kellog &  Deutsche Bank\\\midrule
 $0.05 \, S_{\sigma_0^{i}}$ & 22.26   & 69.68   \\
 $0.10 \, S_{\sigma_0^{i}}$ & 147.32  & 10.09   \\
 $0.15 \, S_{\sigma_0^{i}}$ & 155.65  & 4.03    \\
 $0.20 \, S_{\sigma_0^{i}}$ & 0.05    & 10.31   \\
 $0.25 \, S_{\sigma_0^{i}}$ & 0.11    & 4.82    \\\bottomrule \\
\end{tabular}
\caption{
\emph{Gains per traded assets (GPTA)}
 for the $\G^{\text{fin}}$-arbitrage, applied to historical stock data of the Kellogg Company and Deutsche Bank AG from the year 2000 to 2017. Drift and volatility were estimated by maximum-likelihood methods with a rolling window of length 3 years.}\label{Tab.II.13}
\end{center}
\end{table}
The results confirm the findings from the previous section in the sense that we see gains for all chosen boundaries. If the boundary is chosen too small or too large the trading strategy does, however, not perform optimally. 
\end{example}

\begin{example}[Deutsche Bank]
As a second example, we apply our methodology to stock prices of Deutsche Bank from January 1, 2000 to December 31, 2017. In contrast to the previous example, we observe higher volatility and also large losses in the observation period. 
\begin{figure}[t]
  \centering
  \includegraphics[height=7cm, width = 10.0cm]{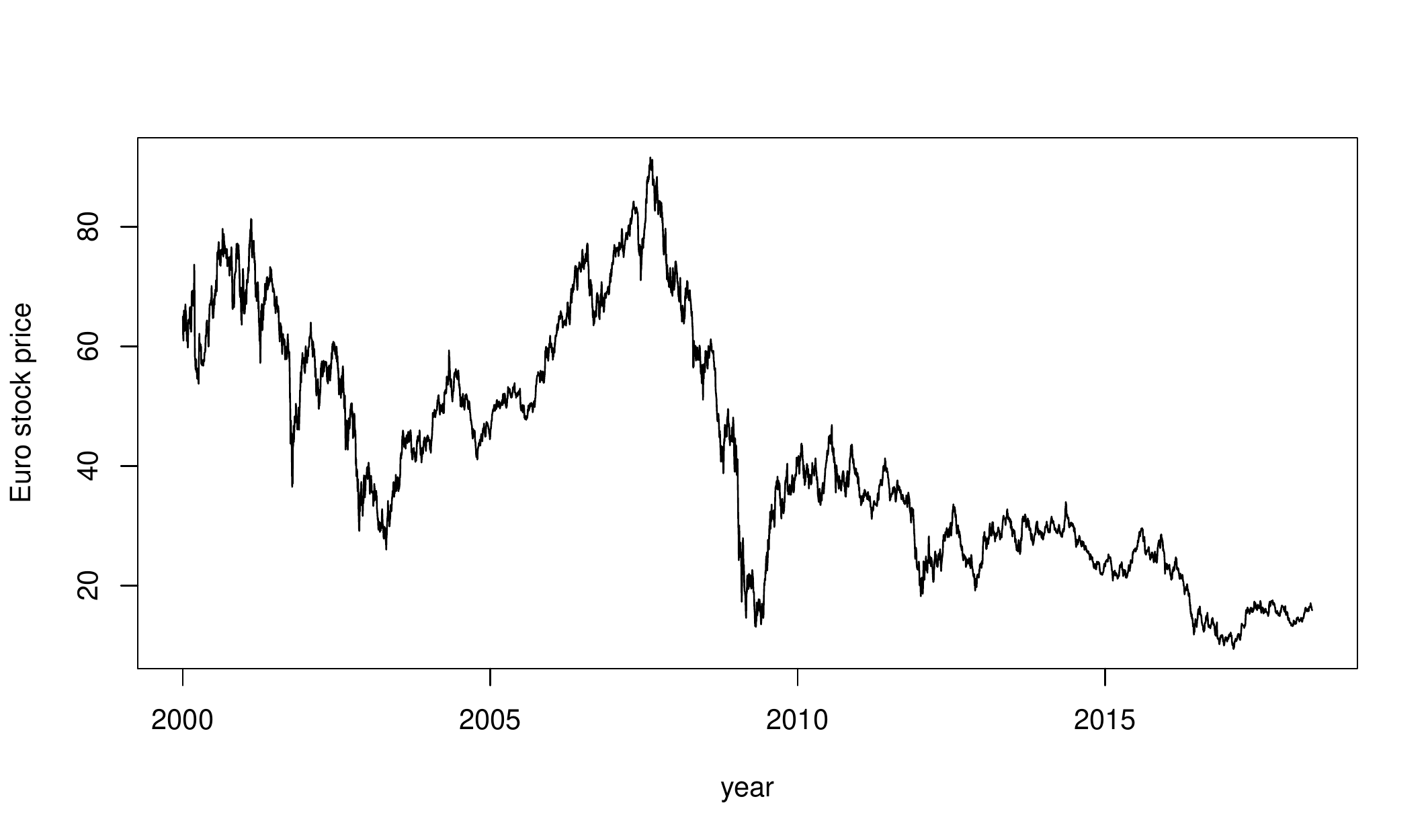}
  \caption{Daily closing prices of the shares of the Deutsche Bank AG during January 1, 2000 and December 31, 2017. Prices are presented in Euro.}\label{Fig.II.9}
\end{figure}
We proceed as for the Kellogg's example and the results are shown in Table \ref{Tab.II.13}.
Due to the present downward trend in the stock evolution the $\G^{\text{fin}}$-strategy is expected to perform as the embedded binomial strategy. We recognize positive gains through all boundaries. 
\end{example}

\section{Conclusion}

We introduce the concept of statistical $\G$-arbitrage and give a characterization of it. Moreover, we 
examine various profitable strategies both on simulated and on market data.
The choice of the information system $\G$ is either motivated naturally by the aim to generate profitable
strategies in average over certain pre-determined scenarios or, alternatively, it can be used as a technical
tool to generate profitable strategies.%

Our data experiments show that the analysed strategies show  a  good  performance 
both on simulated data and on market data.

\begin{appendix}

\section{Proofs}
\begin{proof}[Proof of Lemma \ref{CElemma1}]
 Note that  equations (\ref{CEeq1}) - (\ref{CEeq4}) reads $A\xi \geq 0$  with
 \begin{equation*}
 A =
  \begin{pmatrix}
   \Delta S_1(\omega_2)         & \Delta S_2(\omega_2)    & 0 \\
   \Delta S_1(\omega_6)         & 0       & \Delta S_2(\omega_6) \\
   \Delta S_1(\omega_1) \nu_1 + \Delta S_1(\omega_4)  & \Delta S_2(\omega_1) \nu_1  & \Delta S_2(\omega_4)\\
   \Delta S_1(\omega_3) \nu_2 + \Delta S_1(\omega_5)  & \Delta S_2(\omega_3) \nu_2  & \Delta S_2(\omega_5)
  \end{pmatrix}.
 \end{equation*}
 We do a change of basis for the mapping $A$ and substitute the vector in the first column. This leads to a matrix $\tilde{A}$,    
 \begin{equation*}
 \tilde{A} =
  \begin{pmatrix}
   0  & \Delta S_2(\omega_2)    & 0 \\
   0  & 0       & \Delta S_2(\omega_6) \\
   B_1  & \Delta S_2(\omega_1) \nu_1  & \Delta S_2(\omega_4)\\
   B_2  & \Delta S_2(\omega_3) \nu_2  & \Delta S_2(\omega_5)
  \end{pmatrix}
 \end{equation*}
 where
 \begin{align*}
  B_1 = \nu_1 \left(\Delta S_1(\omega_1) - \Delta S_2(\omega_1) \frac{\Delta S_1(\omega_2)}{\Delta S_2(\omega_2)} \right) + \Delta S_1(\omega_4)- \Delta S_2(\omega_4) \frac{\Delta S_1(\omega_6)}{\Delta S_2(\omega_6)}\\
  B_2 = \nu_2 \left(\Delta S_1(\omega_3) - \Delta S_2(\omega_3) \frac{\Delta S_1(\omega_2)}{\Delta S_2(\omega_2)} \right) + \Delta S_1(\omega_5)- \Delta S_2(\omega_5) \frac{\Delta S_1(\omega_6)}{\Delta S_2(\omega_6)}.
 \end{align*}
 We denote by $\Im(\tilde A)$ the image of a mapping $\tilde A$. There exists statistical arbitrage if 
 $$ \Im(\tilde{A}) \cap \R^4_{>0} \neq \emptyset. $$ 
 The linear subspace spanned by $\tilde{A}$ is given by
 \begin{equation}\label{CEeq5.1}
  \alpha
  \begin{pmatrix}
   0 \\ 0 \\ B_1 \\ B_2
  \end{pmatrix}
  + \beta
  \begin{pmatrix}
   \Delta S_2(\omega_2) \\ 0 \\ \Delta S_2(\omega_1) \nu_1 \\ \Delta S_2(\omega_3) \nu_2
  \end{pmatrix}
  + \gamma 
  \begin{pmatrix}
   0 \\ \Delta S_2(\omega_6) \\ \Delta S_2(\omega_4)\\ \Delta S_2(\omega_5)
  \end{pmatrix}
  ,
 \end{equation}
  with $\alpha, \beta, \gamma \in \R$. Assume this space meets $\R^4_{\geq 0}$. 
  Then it follows from the condition $\beta \Delta S_2(\omega_2) = \beta (s_2^{++}-s_1^+) \geq 0$ that $\beta \geq 0$. Similarily, $\gamma \leq 0$ because $\Delta S_2(\omega_6) = s_2^{--}-s_1^-< 0$. Summing up the third and fourth coordinate from (\ref{CEeq5.1}) we get
  \begin{align}\label{CEeq6}
   &\alpha \left( \nu_1 \Big(\Delta S_1(\omega_1) - \Delta S_2(\omega_1) \frac{\Delta S_1(\omega_2)}{\Delta S_2(\omega_2)} \Big) + \nu_2 \Big(\Delta S_1(\omega_3) - \Delta S_2(\omega_3) \frac{\Delta S_1(\omega_2)}{\Delta S_2(\omega_2)} \Big)\right. \nonumber \\ 
   &\hspace*{3ex}\left. + \frac{\Delta S_1(\omega_6)}{\Delta S_2(\omega_6)} \Big(-\Delta S_2(\omega_4) - \Delta S_2(\omega_5)\Big) + \Delta S_1(\omega_4) + \Delta S_1(\omega_5) \right) \\
   &\hspace*{3ex} + \gamma \left(\Delta S_2(\omega_4) + \Delta S_2(\omega_5) \right) \nonumber \\
   &\hspace*{3ex} +\beta \left(\Delta S_2(\omega_1) \nu_1 + \Delta S_2(\omega_3) \nu_2 \right).   \nonumber
  \end{align}
Choosing $\nu_1 = -\frac{\Delta S_2(\omega_3)}{\Delta S_2(\omega_1)} \nu_2$,  
$$ \beta \left(\Delta S_2(\omega_1) \nu_1 + \Delta S_2(\omega_3) \nu_2 \right)  = 0 $$
such that the last term in the above equation vanishes. 
  As we assumed that the space spanned by (\ref{CEeq5.1}) meets $\R^4_{\geq 0}$ it must also hold true that (\ref{CEeq6}) $\geq 0$. For 
  \begin{equation*}
   \nu_2 < \frac{\frac{\Delta S_1(\omega_6)}{\Delta S_2(\omega_6)} (\Delta S_2(\omega_4) + \Delta S_2(\omega_5)) - \Delta S_1(\omega_4) - \Delta S_1(\omega_5)}{\Delta S_1(\omega_3) - \Delta S_1(\omega_1)\frac{\Delta S_2(\omega_3)}{\Delta S_2(\omega_1)}} = \Gamma_2
  \end{equation*}
 the coefficient of $\alpha$ in (\ref{CEeq6}) is negative. Together with $\gamma \leq 0$ and $\Delta S_2 (\omega_4), \Delta S_2 (\omega_5) > 0$ by assumption this choice of $\nu_2$ results in $\alpha \leq 0$ in order to obtain $(\ref{CEeq6}) \geq 0 $. On the other hand, if we claim
 \begin{equation*}
  \nu_2 > \frac{-\Delta S_1(\omega_5) + \Delta S_2(\omega_5) \frac{\Delta S_1(\omega_6)}{\Delta S_2(\omega_6)}}{\Delta S_1(\omega_3) - \Delta S_2(\omega_3)\frac{\Delta S_1(\omega_2)}{\Delta S_2(\omega_2)}} = \Gamma_1  
 \end{equation*}
  it follows that $B_2 > 0$ and it results for the fourth coordinate of (\ref{CEeq5.1}) that
  \begin{equation*}
   \alpha B_2 + \beta \Delta S_2(\omega_3) \nu_2 + \gamma \Delta S_2(\omega_5) \leq 0.
  \end{equation*}
  Hence $\Im(\tilde{A}) \cap \R^4_{>0} = \emptyset$. 
  It remains to prove that 
  \begin{itemize}
   \item [(i)] $\Gamma_1 < \Gamma_2$ and
   \item [(ii)] there is no statistical arbitrage for $\nu_2 = \Gamma_2$.
  \end{itemize}
  The statements (i) and (ii) are  verified by analogous calculations  which concludes the proof.
\end{proof}

\begin{proof}[Proof of Proposition \ref{Lem.II.1}]
  ``$\Rightarrow$'' If $\det(A) \neq 0$ we choose for example $\xi := A^{-1}1$ and have found an arbitrage opportunity.\par
  ``$\Leftarrow$'' On the other hand, if $\det(A) = 0$ there still might be an arbitrage opportunity if the image of $A$ intersects with the positive subspace of $\R^3$, i.e. if $\Im(A) \cap \R^3_{>0} \neq \emptyset$. 
 To show that this is not the case we change the basis for the mapping $A$ and substitute the vector in the first column. This leads to a matrix $\tilde{A}$,    
 \begin{equation*}
 \tilde{A} =
  \begin{pmatrix}
   0  & \Delta S_2(\omega_1)    & 0 \\
   0  & 0       & \Delta S_2(\omega_4) \\
   B  & \Delta S_2(\omega_2) q  & \Delta S_2(\omega_3)
  \end{pmatrix},
 \end{equation*}
 where 
 \begin{equation*}
 B = q\big(\Delta S_1(\omega_1) - \frac{\Delta S_1(\omega_1)}{\Delta S_2(\omega_1)} \Delta S_2(\omega_2)\big) + \Delta S_1(\omega_3) - \frac{\Delta S_1(\omega_3)}{\Delta S_2(\omega_4)} \Delta S_2(\omega_3).
 \end{equation*}
Calculating $\det(A)$ we see that $\det(A) = 0$ is equivalent to
\begin{align}\label{Eq.II.3.9}
\begin{aligned}
 0 = &-\Delta S_2(\omega_1)\big(\Delta S_1(\omega_4) \Delta S_2(\omega_3) - \big(\Delta S_1(\omega_3) + q \Delta S_1(\omega_2)\big)\Delta S_2(\omega_4)\big) \\
  &  - q \Delta S_1(\omega_1)\Delta S_2(\omega_2)\Delta S_2(\omega_4).\end{aligned}
\end{align}
In the recombining binomial model this reduces to
\begin{equation*}
 0 = q \Delta S_1(\omega_1) \Big(1 - \frac{\Delta S_2(\omega_2)}{\Delta S_2(\omega_1)} \Big) + \Delta S_1(\omega_3) \Big(1 - \frac{\Delta S_2(\omega_3)}{\Delta S_2(\omega_4)} \Big)
\end{equation*}
which is equivalent to $B = 0$. In this case the linear subspace spanned by $\tilde{A}$ is given by
 \begin{equation}\label{CEeq5}
  \alpha
  \begin{pmatrix}
   \Delta S_2(\omega_1) \\ 0 \\ q \Delta S_2(\omega_2)
  \end{pmatrix}
  + \beta
  \begin{pmatrix}
   0 \\ \Delta S_2(\omega_4) \\ \Delta S_2(\omega_3) 
  \end{pmatrix}
  ,
 \end{equation}
 with $\alpha, \, \beta \in \R$. Because $\Delta S_2(\omega_1) > 0$ we need $\alpha \geq 0$ to have arbitrage opportunities. Similar we need to have $\beta \leq 0$ because of $\Delta S_2(\omega_4) < 0$ by assumption. But, as $\Delta S_2(\omega_2) < 0$ and $\Delta S_2(\omega_3)>0$, we obtain for the third coordinate that 
 \begin{equation*}
  \alpha q \Delta S_2(\omega_2) + \beta \Delta S_2(\omega_3) \leq 0
 \end{equation*}
  and hence $\Im(A) \cap \R^3_{>0} = \emptyset$, which concludes the proof. 
\end{proof}

\end{appendix}

\end{document}